\documentclass[fleqn]{llncs}
\usepackage{ifdraft}
\usepackage{amsmath}
\usepackage{amssymb}
\usepackage{color}
\usepackage{prooftree}
\usepackage{stmaryrd}
\usepackage{tikz}
\usepackage{url}
\usepackage{xspace}
\usepackage{booktabs}
\usepackage{times}
\usetikzlibrary{arrows,automata,shapes,snakes}

\usepackage{refcount}

\newcommand{\qedhere}{\qed}

\newif\ifreport
\reporttrue

\newcommand{\reportonly}[1]{
  \ifreport
  #1
  \fi
}
\newcommand{\paperonly}[1]{
  \ifreport
  \else
  #1
  \fi
}

\usepackage{paralist}

\usepackage{thm-restate}

\usepackage[final]{listings}
\definecolor{darkgreen}{rgb}{0,0.5,0}
\lstdefinelanguage{mCRL2}
{
 keywords={act,var,cons,end,eqn,glob,init,val,whr,sort,map,pbes,proc,struct},
 keywords=[2]{true,false,delta,tau},
 keywords=[3]{Bool,Nat,Real,Pos,Int,Set,Bag,List,Int2Nat,Pos2Nat,Int2Pos,min,max},
 keywords=[4]{hide,if,rename,sum,in,mu,nu,forall,exists,mod,allow,block,comm},
 keywords=[5]{nested,initial,state},
 numberstyle=\color{blue},
 comment=[l]\%,
 commentstyle=\itshape\color{darkgreen},
 keywordstyle=[1]\bfseries,
 keywordstyle=[2]\itshape,
 keywordstyle=[3]\itshape,
 keywordstyle=[4]\itshape,
 keywordstyle=[5]\bfseries\itshape,
 basicstyle=\ttfamily\small,
 flexiblecolumns=false
}
\input{macros.tex}

\newtheorem{assumption}{Assumption}

\usepackage{soul}

\begin{document}

\ifreport
\title{Improved Static Analysis for\\ Parameterised Boolean Equation Systems\\
  using Control Flow Reconstruction}
\else
\title{Liveness Analysis for\\ Parameterised Boolean Equation Systems}
\fi
\author{Jeroen J. A. Keiren\inst{1} \and Wieger Wesselink\inst{2} \and Tim A. 
  C. Willemse\inst{2}}
\institute{VU University Amsterdam, The Netherlands\\
  \email{j.j.a.keiren@vu.nl}
  \and
  Eindhoven University of Technology, The Netherlands\\
  \email{\{j.w.wesselink, t.a.c.willemse\}@tue.nl}
}

\maketitle

\begin{abstract}
  We present a sound static analysis technique for fighting the 
  combinatorial explosion
  of parameterised Boolean equation systems (PBESs). These essentially
  are systems of mutually recursive fixed point equations ranging over
  first-order logic formulae.
  Our method detects parameters that are not live by analysing
    a control
    flow graph of a PBES, and it subsequently eliminates such parameters. We show 
    that a naive approach to constructing a control flow graph, needed for the analysis,
    may suffer from an 
    exponential blow-up,
    and we define
    an approximate analysis that avoids this problem.
  The effectiveness of our techniques is
  evaluated using a number of case studies.
\end{abstract}

\section{Introduction}

\emph{Parameterised Boolean equation systems (PBESs)}~\cite{GW:05tcs} are 
systems of fixpoint equations that range over first-order formulae;
they are essentially an equational variation of \emph{Least Fixpoint
  Logic (LFP)}. Fixpoint logics such as PBESs have applications in
database theory and computer aided verification.  For instance, the
CADP~\cite{Gar+:11} and mCRL2~\cite{Cra+:13} toolsets use PBESs 
for model checking and equivalence checking and in \cite{AFJV:11} PBESs are
used to solve Datalog queries.

In practice, the predominant problem for PBESs is evaluating (henceforth 
referred to as \emph{solving}) them so as to 
  answer the decision problem encoded
in them.  There are a variety of techniques for solving
PBESs, see~\cite{GW:05tcs}, but the most straightforward method is
by instantiation to a \emph{Boolean equation system (BES)}
\cite{Mad:97}, and then solving this BES.  
This process is similar to the explicit generation
of a behavioural state space from its symbolic description, and it
suffers from a combinatorial explosion that is akin to the state
space explosion problem.  Combatting this combinatorial
explosion is therefore instrumental in speeding up the process of
solving the problems encoded by PBESs.

While several static analysis techniques have been described using fixpoint
logics, see \eg \cite{CC:77}, with the exception of the static analysis
  techniques for PBESs, described in~\cite{OWW:09}, no such techniques seem to 
  have been applied to fixpoint 
  logics themselves.\medskip

Our main contribution in this paper is a static analysis method for PBESs that
significantly improves over the aforementioned 
techniques for simplifying
PBESs.
In our method, we construct a \emph{control flow graph} 
(CFG) for
  a given PBES and subsequently apply state space reduction 
techniques~\cite{FBG:03,YG:04}, combined with liveness analysis 
techniques
from compiler technology~\cite{ASU:86}.
These typically scrutinise syntactic
descriptions of behaviour to detect and eliminate 
variables that at some point become irrelevant (dead, not live) to the 
behaviour,
thereby decreasing the complexity. 

The notion of control flow of a PBES is not self-evident: formulae
in fixpoint logics (such as PBESs) do not have a notion of a program
counter.  Our notion of control flow is based on the concept of 
\emph{control flow parameters} (CFPs), which induce a CFG.
Similar notions exist in the context of state space
exploration, see~\eg~\cite{PT:09atva}, but so far, no such concept exists
for fixpoint logics. 

The size of the CFGs is potentially exponential
in the number of CFPs. We therefore also describe a modification of 
our analysis---in which reductive power is traded against a lower
complexity---that does not suffer from this problem.
Our static analysis technique allows for solving PBESs using
instantiation that hitherto could not be solved this way, either
because the underlying BESs would be infinite or they would be extremely large.
We show that our methods are 
sound; \ie, simplifying PBESs using our analyses lead to PBESs with the
same solution.

Our static analysis techniques have been implemented in the 
mCRL2 toolset~\cite{Cra+:13} and applied to a set of model
checking and equivalence checking problems. Our experiments show that the
implementations
outperform existing static analysis techniques for PBESs~\cite{OWW:09} in
terms of reductive power, and that reductions of almost 100\% of the size of the
underlying BESs can be achieved. Our experiments confirm that the optimised
version sometimes achieves slightly less reduction 
than our
  non-optimised version, but is faster.
Furthermore, in cases where no additional reduction is achieved compared
to existing techniques, the overhead is mostly 
neglible.

\paragraph{Structure of the paper.}
In Section~\ref{sec:preliminaries} we give a cursory overview of basic
PBES theory and in Section~\ref{sec:example}, we present an example to
illustrate  the difficulty of using instantiation to
solve a PBES and to sketch our solution.  In
Section~\ref{sec:CFP_CFG} we describe our construction of control flow graphs
for PBESs and in Section~\ref{sec:dataflow} we describe our  live parameter
analysis.  We present an optimisation of
the analysis in Section~\ref{sec:local}. The approach is evaluated 
in Section~\ref{sec:experiments}, and 
Section~\ref{sec:conclusions} concludes.
\paperonly{\textbf{We refer to~\cite{KWW:13report} for 
proofs and additional results.}}

\section{Preliminaries}\label{sec:preliminaries}
%
Throughout this paper, we work in a setting of \emph{abstract data
types} with non-empty data sorts $\sort{D_1}, \sort{D_2}, \ldots$,
and operations on these sorts, and a set $\varset{D}$ of sorted
data variables.  We write vectors in boldface, \eg $\var{d}$ is
used to denote a vector of data variables. We write $\var[i]{d}$
to denote the $i$-th element of a vector $\var{d}$.

A semantic set $\semset{D}$ is associated to every sort $\sort{D}$,
such that each term of sort $\sort{D}$, and all operations on
$\sort{D}$ are mapped to the elements and operations of $\semset{D}$
they represent.  \emph{Ground terms} are terms that do not contain
data variables.  For terms that contain data variables, we use an
environment $\dataenv$ that maps each variable from $\varset{D}$
to a value of the associated type.  We assume an interpretation
function $\sem{\_}{}{}$ that maps every term $t$ of sort $\sort{D}$
to the data element $\sem{t}{}{\dataenv}$ it represents, where the
extensions of $\dataenv$ to open terms and vectors are standard. Environment
updates are denoted $\dataenv[\subst{d}{v}]$, where
$\dataenv[\subst{d}{v}](d') = v$ if $d' = d$, and $\dataenv(d')$
otherwise.

We specifically assume the existence of a sort $\sort{B}$ with
elements $\true$ and $\false$ representing the Booleans $\semset{B}$
and a sort $\sort{N} = \{0, 1, 2, \ldots \}$ representing the natural
numbers $\semset{N}$. For these sorts, we assume that the usual operators
are available and, for readability, these are written the same
as their semantic counterparts.

\emph{Parameterised Boolean equation systems}~\cite{Mat:98} are 
sequences of
fixed-point equations ranging over \emph{predicate formulae}. The
latter are first-order formulae extended with predicate
variables, in which the non-logical symbols are taken from the data
language.

\begin{definition}
\label{def:formula}
\label{def:semFormula}
\emph{Predicate formulae} are defined through the following grammar:
$$
\varphi, \psi ::= b \mid X(\val{e}) \mid \varphi \land \psi \mid \varphi \lor
\psi \mid \forall d \colon D. \varphi \mid \exists d \colon D. \varphi$$
in which $b$ is a data term of sort $\sort{B}$, $X(\val{e})$ is a \emph{predicate
variable instance} (PVI) in which $X$ is a predicate variable of
sort $\vec{\sort{D}} \to \sort{B}$, taken from some sufficiently large set
$\varset{P}$ of predicate variables, and $\val{e}$ is a vector of data terms of
sort $\vec{\sort{D}}$.
The interpretation of a predicate formula $\varphi$ in the
context of a predicate
environment $\predenv \colon \varset{P} \to \semset{D} \to \semset{B}$
and data environment $\dataenv$ is denoted as
$\sem{\varphi}{\predenv}{\dataenv}$, where:
\[
\begin{array}{ll}
\sem{b}{\predenv\dataenv}               
  &\isdef 
   \left \{
     \begin{array}{ll} \text{true} & \text{if $\dataenv(b)$ holds} \\
                       \text{false} & \text{otherwise}
     \end{array}
   \right . \\[5pt]
\sem{X(\val{e})}{\predenv\dataenv}               
  &\isdef 
   \left \{
     \begin{array}{ll} \text{true}& \text{if $\predenv(X)(\dataenv(\val{e}))$ 
     holds} \\
                       \text{false} & \text{otherwise}
     \end{array}
   \right . \\[5pt]

\sem{\phi \land \psi}{\predenv\dataenv} 
  &\isdef \sem{\phi}{\predenv\dataenv} \text{ and } \sem{\psi}{\predenv\dataenv} \text{ hold} \\[5pt]
  \sem{\phi \lor \psi}{\predenv\dataenv}
  &\isdef \sem{\phi}{\predenv\dataenv} \text{ or } \sem{\psi}{\predenv\dataenv} \text{ hold} \\[5pt]
\sem{\forall{d \colon \sort{D}}.~ \phi}{\predenv\dataenv}
  &\isdef \text{for all ${v \in \semset{D}}$, }~\sem{\phi}{\predenv\dataenv[v/d]} \text{ holds} \\[5pt]
\sem{\exists{d \colon \sort{D}}.~ \phi}{\predenv\dataenv} 
  &\isdef \text{for some ${v \in \semset{D}}$, }~\sem{\phi}{\predenv\dataenv[v/d]} \text{ holds}
\end{array}
\]
\end{definition}
We assume the usual precedence rules for the logical operators.
\emph{Logical equivalence} between two predicate formulae $\varphi,
\psi$, denoted $\varphi \equiv \psi$, is defined as 
$\sem{\varphi}{\predenv\dataenv}
= \sem{\psi}{\predenv\dataenv}$ for all $\predenv, \dataenv$.
Freely 
occurring data variables
in $\varphi$ are denoted by $\free{\varphi}$. We refer to $X(\val{e})$ occuring
in a predicate formula as a \emph{predicate variable instance} (PVI).
For simplicity, we assume that if a data variable is bound by a quantifier
in a formula $\varphi$, it does not also occur free within $\varphi$.
\begin{definition}
\label{def:PBES}
PBESs are defined by the following grammar:
$$
\PBES ::= \emptyPBES \mid (\nu X(\var{d} \colon \vec{D}) = \varphi) \PBES
                   \mid (\mu X(\var{d} \colon \vec{D}) = \varphi) \PBES
$$
in which $\emptyPBES$ denotes the empty equation system; $\mu$ and $\nu$ are the
least and greatest fixed point signs, respectively; $X$ is a sorted predicate
variable of sort $\vec{\sort{D}} \to \sort{B}$, $\var{d}$ is a vector of formal 
parameters,
and $\varphi$ is a predicate formula. We henceforth omit a trailing $\emptyPBES$.
\end{definition}
By convention $\rhs{X}$ denotes the right-hand side of the
defining equation for $X$ in a PBES $\PBES$;
$\param{X}$ denotes the set of \emph{formal parameters} of $X$ and
we assume that $\free{\rhs{X}} \subseteq
\param{X}$. By superscripting a formal parameter with the predicate
variable to which it belongs, we distinguish between formal parameters for
different predicate variables, \ie, we 
write $d^X$ when $d \in \param{X}$. We write $\sigma$ to stand for
either $\mu$ or $\nu$.
 
The set of \emph{bound predicate variables} of some PBES $\PBES$, denoted
$\bnd{\PBES}$, is the set of predicate variables occurring
at the left-hand sides of the equations in $\PBES$.  Throughout this
paper, we deal with PBESs that are both \emph{well-formed}, \ie for every
$X \in \bnd{\PBES}$ there is exactly one equation in $\PBES$, and
\emph{closed}, \ie for every $X \in \bnd{\PBES}$, only predicate variables taken
from $\bnd{\PBES}$ occur in $\rhs{X}$.

To each PBES $\PBES$ we associate a \emph{top assertion}, denoted
$\mcrlKw{init}~X(\val{v})$, where we require $X \in \bnd{\PBES}$. For
a parameter $\var[m]{d} \in \param{X}$ for the top assertion
$\mcrlKw{init}~X(\val{v})$ we define the value $\init{\var[m]{d}}$
as $\val[m]{v}$.\\

We next define a PBES's semantics.  Let $\semset{B}^{\vec{\semset{D}}}$
denote the set of functions $f \colon \vec{\semset{D}} \to \semset{B}$,
and define the ordering $\sqsubseteq$ as $f \sqsubseteq g$ iff for
all $\vec{v} \in \vec{\semset{D}}$, $f(\vec{v})$ implies $g(\vec{v})$.
For a given pair of environments $\dataenv, \predenv$, a predicate 
formula $\varphi$ gives rise to a predicate transformer $T$
on the complete lattice 
$(\semset{B}^{\vec{\semset{D}}}, \sqsubseteq)$ as follows:
$
T(f) = \lambda \vec{v} \in \vec{\semset{D}}. 
\sem{\varphi}{\predenv[\subst{X}{f}]}{\dataenv[\subst{\vec{d}}{\vec{v}}]}
$.

Since the predicate transformers defined this way are monotone,
their extremal fixed points exist. We denote the least fixed point of
a given predicate transformer $T$ by $\mu T$, and the greatest fixed point
of $T$ is denoted $\nu T$. 
\begin{definition}
The \emph{solution} of an equation system in the context of a predicate
environment $\predenv$ and data environment $\dataenv$ is defined inductively
as follows:
\begin{align*}
\sem{\emptyPBES}{\predenv}{\dataenv} & \isdef \predenv \\
\sem{(\mu X(\var{d} \colon \vec{D}) = \rhs{X}) \PBES}{\predenv}{\dataenv}
& \isdef \sem{\PBES}{\predenv[\subst{X}{\mu T}]}{\dataenv}\\
\sem{(\nu X(\var{d} \colon \vec{D}) = \rhs{X}) \PBES}{\predenv}{\dataenv}
& \isdef \sem{\PBES}{\predenv[\subst{X}{\nu T}]}{\dataenv}
\end{align*}
with $T(f) = \lambda \val{v} \in \val{\semset{D}}. 
\sem{\varphi}{(\sem{\PBES}{\predenv[\subst{X}{f}]}{\dataenv})}
{\dataenv[\subst{\var{d}}{\val{v}}]}$
\end{definition}
The solution prioritises the fixed point signs of left-most equations
over the fixed point signs of equations that follow, while respecting
the equations. Bound predicate variables of closed PBESs have a
solution that is independent of the predicate and data environments
in which it is evaluated.  We therefore omit these environments and
write $\sem{\PBES}(X)$ instead of $\sem{\PBES}{\predenv}{\dataenv}(X)$.

\reportonly{
\newcommand{\fixp}{\sigma}
\newcommand{\rankE}[2]{\ensuremath{\mathsf{rank}_{#1}(#2)}}
\newcommand{\rank}[1]{\ensuremath{\mathsf{rank}(#1)}}

The \emph{signature} \cite{Wil:10} of a predicate variable $X$ of sort
$\vec{\sort{D}} \to \sort{B}$, $\signature{X}$, is the product $\{X\} \times
\semset{D}$.
The notion of signature is lifted to sets of predicate variables $P \subseteq
\varset{P}$ in the natural way, \ie
$\signature{P} = \bigcup_{X \in P} \signature{X}$.\footnote{Note that in
\cite{Wil:10} the notation $\mathsf{sig}$ is used to denote the signature. Here
we deviate from this notation due to the naming conflict with the
\emph{significant parameters} of a formula, which also is standard notation
introduced in \cite{OWW:09}, and which we introduce in 
Section~\ref{sec:dataflow}.}
\begin{definition}[{\cite[Definition~6]{Wil:10}}]\label{def:r-correlation}
Let ${\rel{R}} \subseteq \signature{\varset{P}} \times \signature{\varset{P}}$ 
be 
an arbitrary
relation. A predicate environment $\predenv$ is an $\rel{R}$-correlation iff
$(X, \val{v}) {\rel{R}} (X', \val{v'})$ implies $\predenv(X)(\val{v}) =
\predenv(X')(\val{v'})$.
\end{definition}
A \emph{block} is a non-empty equation system of like-signed fixed point
equations. Given an equation system $\PBES$, a block $\mathcal{B}$ is maximal
if its neighbouring equations in $\PBES$ are of a different sign than the
equations in $\mathcal{B}$. The $i^\mathit{th}$ maximal block in $\PBES$ is 
denoted by
$\block{i}{\PBES}$.
For relations $\rel{R}$ we write $\correnv{\rel{R}}$ for the set of
$\rel{R}$-correlations.
\begin{definition}[{\cite[Definition~7]{Wil:10}}]
\label{def:consistent-correlation} Let $\PBES$ be an equation system.
Relation ${\rel{R}} \subseteq \signature{\varset{P}} \times 
\signature{\varset{P}}$ is
a \emph{consistent correlation} on $\PBES$, if for $X, X' \in \bnd{\PBES}$,
$(X, \val{v}) \rel{R} (X', \val{v'})$ implies:
\begin{compactenum}
  \item for all $i$, $X \in \bnd{\block{i}{\PBES}}$ iff $X' \in 
  \bnd{\block{i}{\PBES}}$
  \item for all $\predenv \in \correnv{\rel{R}}$, $\dataenv$, we have
        $\sem{\rhs{X}}{\predenv \dataenv[\subst{\var{d}}{\val{v}}]} = 
        \sem{\rhs{X'}}{\predenv \dataenv[\subst{\var{d'}}{\val{v'}}]}$
\end{compactenum}
For $X, X' \in \bnd{\PBES}$, we say $(X, \val{v})$ and $(X', \val{v'})$ 
consistently
correlate, denoted as $(X, \val{v}) \consistent (X', \val{v'})$ iff there
exists a correlation $\rel{R} \subseteq
\signature{\bnd{\PBES}} \times \signature{\bnd{\PBES}}$ such that $(X, \val{v})
\rel{R} (X', \val{v'})$ .
\end{definition}
Consistent
correlations can be lifted to variables in different equation systems in $\PBES$
and $\PBES'$, assuming that the variables in the equation systems do not 
overlap.
We call such equation systems \emph{compatible}.
Lifting consistent correlations to different equation systems can, \eg, be
achieved by merging the equation systems to an
equation system $\mathcal{F}$, in which, if $X \in \bnd{\PBES}$, then
$X \in \bnd{\block{i}{\PBES}}$ iff $X \in \bnd{\block{i}{\mathcal{F}}}$,
and likewise for $\PBES'$.
The consistent correlation can then be defined on $\mathcal{F}$.

The following theorem \cite{Wil:10} shows the relation between consistent
correlations and the solution of a PBES.
\begin{theorem}[{\cite[Theorem~2]{Wil:10}}]\label{thm:willemse}\label{thm:cc}
Let $\PBES$, $\PBES'$ be compatible equation systems, and $\consistent$ a
consistent correlation. Then for all $X \in \bnd{\PBES}$,
$X' \in \bnd{\PBES'}$ and all $\predenv \in \correnv{\consistent}$, we have
$(X, \val{v}) \consistent (X', \val{v'}) \implies \sem{\PBES}{\predenv 
\dataenv}(X)(\val{v}) = 
\sem{\PBES'}{\predenv \dataenv}(X')(\val{v'})$
\end{theorem}

We use this theorem in proving the correctness of our static analysis
technique.
} 


\section{A Motivating Example}\label{sec:example}

In practice, solving PBESs proceeds via \emph{instantiating}~\cite{PWW:11}
into \emph{Boolean equation systems (BESs)}, for which solving is
decidable. The latter is the fragment of PBESs with equations
that range over propositions only, \ie, formulae without data and
quantification.  Instantiating a PBES to a BES is akin to state space
exploration and suffers from a similar combinatorial
explosion. Reducing the time spent on it is thus instrumental in speeding
up, or even enabling the solving process.
We illustrate this using the following (academic) example,  which we also
use as our running example:
\[
\begin{array}{lll}
\nu X(i,j,k,l\colon\sort{N}) & = & 
( i \not= 1 \vee j \not= 1 \vee X(2,j,k,l+1)) \wedge \forall m\colon\sort{N}. Z(i,2,m+k,k)   \\
\mu Y(i,j,k,l\colon\sort{N}) & = &
k = 1 \vee (i = 2 \wedge X(1,j,k,l) ) \\
\nu Z(i,j,k,l\colon\sort{N}) & = &
(k < 10 \vee j = 2) \wedge (j \not= 2 \vee Y(1,1,l,1) ) \wedge
Y(2,2,1,l)  
\end{array}
\]
The presence of PVIs $X(2,j,k,l+1)$ and $Z(i,2,m+k,k)$ in $X$'s
equation means the solution to $X(1,1,1,1)$ depends on the
solutions to $X(2,1,1,2)$ and $Z(1,2,v+1,1)$,
for all values $v$, see Fig.~\ref{fig:instantiate}.  Instantiation
finds these dependencies by simplifying the right-hand
side of $X$ when its parameters have been assigned value $1$:
\[
( 1 \not= 1 \vee 1 \not= 1 \vee X(2,1,1,1+1))
 \wedge \forall m\colon\sort{N}. Z(1,2,m+1,1) 
\]
Since for an infinite number of different arguments the solution
to $Z$ must be computed, instantiation does not terminate.  The problem is with
the third parameter ($k$) of $Z$.  We cannot simply assume
that values assigned to the third parameter of $Z$ do not matter;
in fact, only when $j =2$, $Z$'s right-hand side predicate formula
does not depend on $k$'s value.  This is where our developed method will
come into play: it automatically
determines that it is sound to replace PVI
$Z(i,2,m+k,k)$ by, \eg, $Z(i,2,1,k)$ and to remove the universal
quantifier, enabling us to solve $X(1,1,1,1)$ using
instantiation.

Our technique uses a \emph{Control Flow Graph} (CFG) underlying the
PBES for analysing which parameters of a PBES are \emph{live}.  
The CFG is a finite abstraction of the dependency graph
that would result from instantiating a PBES. For instance, when
ignoring the third and fourth parameters in our example PBES, 
we find that the solution to $X(1,1,*,*)$ depends
on the first PVI, leading to $X(2,1,*,*)$ and the second PVI in $X$'s
equation, leading to $Z(1,2,*,*)$. In the same way we can determine
the dependencies for $Z(1,2,*,*)$, resulting in the finite structure
depicted in Fig.~\ref{fig:CFG}. The subsequent liveness analysis annotates
each vertex with a label indicating which parameters 
cannot (cheaply) be excluded from
having an impact on the solution to the equation system; these are
assumed to be live. Using these labels, we modify the PBES automatically.

  \begin{figure}[t]
  \centering
    \parbox[b]{.33\textwidth}
    {
    \begin{tikzpicture}[>=stealth',draw,node distance=30pt,inner 
    sep=0pt]
    \tiny
    \node[state,shape=rectangle,draw=none] (X1111) {$X(1,1,1,1)$};
    \node[state,shape=rectangle,draw=none, below of=X1111,yshift=-30pt] (X2112) {$X(2,1,1,2)$};
    \node[state,shape=rectangle,draw=none, right of=X1111,xshift=20pt] (Z1211) {$Z(1,2,1,1)$};
    \node[state,shape=rectangle,draw=none, right of=Z1211,xshift=10pt] (Y2211) {};

    \node[state,shape=rectangle,draw=none, below of=Z1211,yshift=17pt] (Z1221) {$Z(1,2,2,1)$};
    \node[state,shape=rectangle,draw=none, right of=Z1221,xshift=10pt] (Y2221) {};
    \node[state,shape=rectangle,draw=none, below of=Z1221,yshift=17pt] (Z1231) {$Z(1,2,3,1)$};
    \node[state,shape=rectangle,draw=none, right of=Z1231,xshift=10pt] (Y2231) {};
    \node[state,shape=rectangle,draw=none, below of=Z1231,yshift=22pt] (Z1241) {\phantom{$Z(1,2,4,1)$}};
    \node[state,shape=rectangle,draw=none, below of=Z1241,yshift=24pt] (Z1251) {\phantom{$Z(1,2,5,1)$}};
    
    \node[state,shape=rectangle,draw=none, right of=X2112,xshift=20pt] (Z2211) {$Z(2,2,1,1)$};
    \node[state,shape=rectangle,draw=none, right of=Z2211,xshift=10pt] (Y3211) {};
    \node[state,shape=rectangle,draw=none, below of=Z2211,yshift=17pt] (Z2221) {$Z(2,2,2,1)$};
    \node[state,shape=rectangle,draw=none, right of=Z2221,xshift=10pt] (Y3221) {};
    \node[state,shape=rectangle,draw=none, below of=Z2221,yshift=17pt] (Z2231) {$Z(2,2,3,1)$};
    \node[state,shape=rectangle,draw=none, right of=Z2231,xshift=10pt] (Y3231) {};
    \node[state,shape=rectangle,draw=none, below of=Z2231,yshift=22pt] (Z2241) 
    {\phantom{$Z(2,2,4,1)$}};
    \node[state,shape=rectangle,draw=none, below of=Z2241,yshift=24pt] (Z2251) 
    {\phantom{$Z(2,2,5,1)$}};

    \draw[->]
     (X1111) edge (X2112.north)
     (X1111) edge (Z1211.west)
     (X1111) edge (Z1221.west)
     (X1111) edge (Z1231.west)
     (X1111) edge[dashed] (Z1241.west)
     (X1111) edge[dashed] (Z1251.west)
     (X2112) edge (Z2211.west)
     (X2112) edge (Z2221.west)
     (X2112) edge (Z2231.west)
     (X2112) edge[dashed] (Z2241.west)
     (X2112) edge[dashed] (Z2251.west);
     \draw
     (Z1211) edge[dotted] (Y2211.west)
     (Z1221) edge[dotted] (Y2221.west)
     (Z1231) edge[dotted] (Y2231.west)
     (Z2211) edge[dotted] (Y3211.west)
     (Z2221) edge[dotted] (Y3221.west)
     (Z2231) edge[dotted] (Y3231.west)
     
     ;
    \end{tikzpicture}
    \caption{Dependency graph}
    \label{fig:instantiate}
    }
    \qquad\parbox[b]{.61\textwidth}
    {
    \begin{tikzpicture}[>=stealth',draw,,node distance=45pt,inner 
    sep=0pt,minimum size=17pt]
    \tiny
    \node[state,label = left:{\scriptsize$\{k\}$}] (X11) {$\begin{array}{c}X \\ i=1 \\j= 1\end{array}$};
    \node[state, label=left:{\scriptsize$\{k\}$}, below of=X11,xshift=45pt] (X12) {$\begin{array}{c}X\\ i=1\\ j= 2\end{array}$};
    \node[state,  label=left:{\scriptsize$\{k\}$},below of=X12,xshift=-45pt] (X21)  {$\begin{array}{c}X\\ i=2\\ j= 1\end{array}$};
    \node[state, label=right:{\scriptsize$\{l\}$}, right of=X11,xshift=45pt] (Z12) {$\begin{array}{c}Z\\ i=1\\ j= 2\end{array}$};
    \node[state,label=right:{\scriptsize$\{k\}$}, right of=X12] (Y22)  {$\begin{array}{c}Y\\ i=2\\ j= 2\end{array}$};
    \node[state, label=right:{\scriptsize$\{l\}$}, below of=Y22] (Z22) {$\begin{array}{c}Z\\ i=2\\ j= 2\end{array}$};
    \node[state,label=right:{\scriptsize$\{k\}$}, right of=Y22] (Y11)  {$\begin{array}{c}Y\\ i=1\\ j= 1\end{array}$};

    \path[->]
     (X11) edge node[above] {$2$} (Z12)
     (X11) edge node[left]  {$1$} (X21)
     (X12) edge node[left] {$2$} (Z12)
     (X21) edge node[above] {$2$} (Z22)
     (Y22) edge node[above] {$1$} (X12)
     (Z12) edge node[right] {$1$} (Y11)
     (Z12) edge node[right] {$2$} (Y22)
     (Z22) edge node[right] {$1$} (Y11)
     (Z22) edge node[left] {$2$} (Y22)
    ;
    \end{tikzpicture}
    \caption{Control flow graph for the running example}
    \label{fig:CFG}
    }
  \end{figure}
Constructing a good CFG is a major difficulty, which we address in
Section~\ref{sec:CFP_CFG}. The liveness analysis and the subsequent
modification of the analysed PBES is described in
Section~\ref{sec:dataflow}. Since the CFG constructed in
Section~\ref{sec:CFP_CFG} can still suffer from a combinatorial
explosion, we present an optimisation of our analysis in
Section~\ref{sec:local}.

\section{Constructing Control Flow Graphs for PBESs}
\label{sec:CFP_CFG}

The vertices in the control flow graph we constructed in the previous section
represent the values assigned to a subset of the equations' formal
parameters whereas an edge between two vertices captures the
dependencies among (partially instantiated) equations.  The better
the control flow graph approximates the dependency graph resulting
from an instantiation, the more precise the resulting liveness
analysis.

Since computing a precise control flow graph is expensive,
the problem is to
compute the graph effectively and balance
precision and cost.
To this end, we first identify a set of 
\emph{control flow parameters}; the values to
these parameters will make up the vertices in the control flow
graph. While there is some choice for control flow parameters,
we require that these are parameters for which we can
\emph{statically} determine:
\begin{compactenum}
\item the (finite set of) values these parameters can assume,

\item the set of PVIs on which the truth of a right-hand
side predicate formula may depend, given a concrete value for each control flow
parameter, and

\item the values assigned to the control flow parameters by all
PVIs on which the truth of a right-hand side predicate formula
may depend.

\end{compactenum}
In addition to these requirements, we impose one other restriction:
control flow parameters of one equation must be \emph{mutually
independent}; \ie, we have to be able to determine their values
independently of each other. Apart from being a natural requirement
for a control flow parameter, it enables us to devise optimisations of
our liveness analysis.

We now formalise these ideas. First, we characterise three partial
functions that together allow to relate values of formal parameters
to the dependency of a formula on a given PVI. Our formalisation
of these partial functions is based on the following observation:
if in a formula $\varphi$, we can replace a particular PVI $X(\val{e})$
with the subformula $\psi \wedge X(\val{e})$ without this affecting
the truth value of $\varphi$, we know that $\varphi$'s truth value
only depends on $X(\val{e})$'s whenever $\psi$ holds.  We will
choose $\psi$ such that it allows us to pinpoint exactly what value
a formal parameter of an equation has (or will be assigned through
a PVI). Using these functions, we then identify our
control flow parameters by eliminating variables that do
not meet all of the aforementioned requirements.

In order to reason about individual PVIs occurring in predicate
formulae we introduce the notation necessary to do so.  Let
$\npred{\varphi}$ denote the number of PVIs occurring in a predicate
formula $\varphi$.  The function $\predinstphi{\varphi}{i}$ is the
formula representing the $i^\text{th}$ PVI in $\varphi$, of which
$\predphi{\varphi}{i}$ is the name and $\dataphi{\varphi}{i}$
represents the term that appears as the argument of the instance. In general
$\dataphi{\varphi}{i}$ is a vector, of which we denote the $j^{\text{th}}$
argument by $\dataphi[j]{\varphi}{i}$.
Given predicate formula $\psi$ we write $\varphi[i \mapsto \psi]$
to indicate that the PVI at position $i$ is replaced syntactically
by $\psi$ in $\varphi$.
%
\reportonly{
Formally we define $\varphi[i \mapsto 
\psi]$,
as follows.
\begin{definition}
Let $\psi$ be a predicate formula, and let $i \leq \npred{\varphi}$,
$\varphi[i \mapsto \psi]$ is defined inductively as follows.

\begin{align*}
b[i \mapsto \psi] & \isdef b \\
Y(e)[i \mapsto \psi] & \isdef \begin{cases} \psi & \text{if $i = 1$}\\ Y(e) & 
\text{otherwise} \end{cases} \\
(\forall d \colon D . \varphi)[i \mapsto \psi] & \isdef \forall d \colon D . 
\varphi[i \mapsto \psi]\\
(\exists d \colon D . \varphi)[i \mapsto \psi] & \isdef \exists d \colon D . 
\varphi[i \mapsto \psi]\\
(\varphi_1 \land \varphi_2)[i \mapsto \psi] & \isdef \begin{cases}
\varphi_1 \land \varphi_2[(i - \npred{\varphi_1}) \mapsto \psi] & \text{if } i 
> \npred{\varphi_1} \\
\varphi_1[i \mapsto \psi] \land \varphi_2 & \text{if } i \leq \npred{\varphi_1}
\end{cases}\\
(\varphi_1 \lor \varphi_2)[i \mapsto \psi] & \isdef \begin{cases}
\varphi_1 \lor \varphi_2[(i - \npred{\varphi_1}) \mapsto \psi] & \text{if } i > 
\npred{\varphi_1} \\
\varphi_1[i \mapsto \psi] \lor \varphi_2 & \text{if } i \leq \npred{\varphi_1}
\end{cases}
\end{align*}
\end{definition}
} 
\begin{definition} Let $s \colon \varset{P} \times \mathbb{N} \times \mathbb{N}
\to D$, $t \colon \varset{P} \times \mathbb{N} \times \mathbb{N} \to D$, and
$c \colon \varset{P} \times \mathbb{N} \times \mathbb{N} \to \mathbb{N}$
be partial
functions, where $D$ is the union of all ground terms.
The triple $(s,t,c)$ is a \emph{unicity constraint} for PBES $\PBES$ if for all
$X \in \bnd{\PBES}$,  $i,j,k \in \mathbb{N}$ and
ground terms $e$:
\begin{compactitem}
\item (source) if $s(X,i,j) {=} e$ then
$\rhs{X} \equiv
\rhs{X}[i \mapsto (\var[j]{d} = e \wedge
\predinstphi{\rhs{X}}{i})]$,

\item (target) if $t(X,i,j) {=} e$ then
$\rhs{X}
\equiv \rhs{X}[i \mapsto (\dataphi[j]{\rhs{X}}{i} = e \wedge
\predinstphi{\rhs{X}}{i})]$, 

\item (copy) if $c(X,i,j) {=} k$ then $\rhs{X} \equiv \rhs{X}[i \mapsto
(\dataphi[k]{\rhs{X}}{i} = \var[j]{d} \wedge
\predinstphi{\rhs{X}}{i} )]$.

\end{compactitem}
\end{definition}
Observe that indeed, function $s$ states that, when defined, formal
parameter $\var[j]{d}$ must have value $s(X,i,j)$ for $\rhs{X}$'s
truth value to depend on that of $\predinstphi{\rhs{X}}{i}$.  In
the same vein $t(X,i,j)$, if defined, gives the fixed value of the
$j^\text{th}$ formal parameter of $\predphi{\rhs{X}}{i}$.
Whenever $c(X,i,j) = k$ the value of variable $\var[j]{d}$ is
transparently copied to position $k$ in the $i^\text{th}$ predicate
variable instance of $\rhs{X}$. Since $s,t$ and $c$ are partial
functions, we do not require them to be defined; we use $\bot$ to
indicate this.
\begin{example}\label{exa:unicity_constraint}
  A unicity constraint $(s,t,c)$ for our running example 
  could be one that assigns $s(X,1,2) = 1$, since parameter $j^X$
  must be $1$ to make $X$'s right-hand side formula depend on PVI
  $X(2,j,k,l+1)$. We can set $t(X,1,2) = 1$, as one can deduce that
  parameter $j^X$ is set to $1$ by the PVI $X(2,j,k,l+1)$;
  furthermore, we can set $c(Z,1,4) = 3$, as parameter $k^Y$ is 
  set to $l^Z$'s value by PVI $Y(1,1,l,1)$.

\end{example}

\reportonly{
The requirements allow unicity constraints to be underspecified. In practice,
it is desirable to choose the constraints as complete as possible. If, in a
unicity constraint $(s,t,c)$, $s$ and $c$ are defined for a predicate variable
instance, it can immediately be established that we can define $t$ as well.
This is formalised by the following property.
\begin{property}\label{prop:sourceCopyDest}
Let $X$ be a predicate variable, $i \leq \npred{\rhs{X}}$, let $(s,t,c)$ be
a unicity constraint, and let $e$ be a value, then
$$
(s(X, i, n) = e \land c(X, i, n) = m) \implies t(X, i, m) = e.
$$
\end{property}
Henceforth we assume that all unicity constraints satisfy this property.
The overlap between $t$ and $c$ is now straightforwardly formalised in the
following lemma.
\begin{lemma}\label{lem:copyDest}
Let $X$ be a predicate variable, $i \leq \npred{\rhs{X}}$, and let
$(s,t,c)$ be a unicity constraint, then if
$(s(X, i, n)$ and $t(X, i, m)$ are both defined,
$$
c(X, i, n) = m \implies s(X, i, n) = t(X, i, m).
$$
\end{lemma}
\begin{proof}
Immediately from the definitions and Property~\ref{prop:sourceCopyDest}.
\end{proof}
} 
From hereon, we assume that $\PBES$ is an arbitrary PBES 
with $(\sourcename,\destname,\copiedname)$
a unicity constraint we can deduce for it.
Notice that for each formal parameter for which either \sourcename
or \destname is defined for some PVI, we have a finite set of values
that this parameter can assume.  However, at this point we do not
yet know whether this set of values is exhaustive: it may be that
some PVIs may cause the parameter to take on arbitrary values.
Below, we will narrow down for which parameters we \emph{can}
ensure that the set of values is exhaustive. First, we eliminate
formal parameters that do not meet conditions 1--3 for PVIs that
induce self-dependencies for an equation.
\begin{definition}\label{def:LCFP}
A parameter $\var[n]{d} \in \param{X}$ is a \emph{local
control flow parameter} (LCFP) if for all $i$ such that $\predphi{\rhs{X}}{i}
= X$, either $\source{X}{i}{n}$ and $\dest{X}{i}{n}$ are defined, or
$\copied{X}{i}{n} = n$.

\end{definition}
\begin{example} Formal parameter $l^X$ in our running example
does not meet the conditions of Def.~\ref{def:LCFP} and is therefore not
an LCFP. All other parameters in all other equations are still LCFPs since
$X$ is the only equation with a self-dependency.
\end{example}
From the formal parameters that are LCFPs, we 
next eliminate those parameters that do not meet conditions 1--3
for PVIs that induce dependencies among \emph{different} equations.
\begin{definition}\label{def:GCFP}
A parameter $\var[n]{d} \in \param{X}$
is a \emph{global control flow parameter} (GCFP)
if it is an LCFP, and for all $Y \in \bnd{\mathcal{E}}\setminus \{X\}$ and
all $i$ such that $\predphi{\rhs{Y}}{i} = X$, either 
$\dest{Y}{i}{n}$ is defined, or $\copied{Y}{i}{m} = n$
for some GCFP $\var[m]{d} \in \param{Y}$. 

\end{definition}
The above definition is recursive in nature: if a parameter does
not meet the GCFP conditions then this may result in another parameter
also not meeting the GCFP conditions.  Any set of parameters that
meets the GCFP conditions is a good set, but larger sets possibly lead to
better information about the control flow in a PBES.
\begin{example}
Formal parameter $k^Z$ in our running example
is not a GCFP since in PVI $Z(i,2,m+k,1)$ from $X$'s equation,
the value assigned to $k^Z$ cannot be determined.
\end{example}
The parameters that meet the GCFP conditions satisfy the conditions
1--3 that we imposed on control flow parameters: they assume a
finite set of values, we can deduce which PVIs may affect the
truth of a right-hand side predicate formula, and we can deduce how
these parameters evolve as a result of all PVIs in a PBES. However,
we may still have parameters of a given equation that are mutually
dependent. Note that this dependency can only arise as a result of
copying parameters: in all other cases, the functions \sourcename
and \destname provide the information to deduce concrete values.
\begin{example}
GCFP $k^Y$ affects GCFP $k^X$'s value through PVI $X(1,j,k,l)$; likewise,
$k^X$ affects
$l^Z$'s value through PVI $Z(i,2,m+k,k)$. 
Through the PVI $Y(2,2,1,l)$ in $Z$'s equation,
GCFP $l^Z$ affects GCFPs $l^Y$ value. Thus, $k^Y$ affects $l^Y$'s value
transitively.
\end{example}
We identify parameters that, through copying, may become
mutually dependent. To this end, we use a relation $\sim$, to
indicate that GCFPs are \emph{related}.  Let $\var[\!\!\!\!n]{d^X}$ and
$\var[\!\!\!\!m]{d^Y}$ be GCFPs; these are \emph{related}, denoted 
$\var[\!\!\!\!n]{d^X}
\sim \var[\!\!\!\!m]{d^Y}$, if $n = \copied{Y}{i}{m}$ for some $i$.  Next,
we characterise when a set of GCFPs does not introduce mutual
dependencies.
\begin{definition}
\label{def:control_structure}
Let $\mathcal{C}$ be a set of GCFPs, and let $\sim^*$ denote the reflexive,
symmetric and transitive closure of $\sim$ on $\mathcal{C}$.
Assume ${\approx} \subseteq \mathcal{C} \times \mathcal{C}$ is an equivalence 
relation that subsumes $\sim^*$; \ie,
that satisfies $\sim^* \subseteq \approx$. Then the pair
$\langle \mathcal{C}, \approx \rangle$ defines a \emph{control structure}
if for all $X \in  \bnd{\mathcal{E}}$ and all $d,d' \in \mathcal{C} \cap
\param{X}$, if $d \approx d'$, then $d = d'$.
\end{definition}
We say that a unicity constraint is a \emph{witness} to a control
structure $\langle \varset{C},\approx\rangle$ if the latter can be
deduced from the unicity constraint through 
Definitions~\ref{def:LCFP}--\ref{def:control_structure}. 
The equivalence $\approx$ in a control structure
also serves to identify GCFPs that take on the same role in
\emph{different} equations:  we say that two parameters $c,c' \in
\varset{C}$ are \emph{identical} if $c \approx c'$.
As a last step, we formally define our notion of a
control flow parameter.
\begin{definition}
A formal parameter $c$ is a \emph{control flow parameter (CFP)} if there
is a control structure $\langle \varset{C},\approx\rangle$ such
that $c \in \varset{C}$.
\end{definition}
\begin{example}\label{exa:CFP}
Observe that there is a unicity constraint that
identifies that parameter
$i^X$ is copied to $i^Z$ in our running example. Then necessarily $i^Z \sim 
i^X$ and thus
$i^X \approx i^Z$ for a control structure $\langle \mathcal{C},\approx \rangle$
with $i^X,i^Z \in \mathcal{C}$. 
However, $i^X$ and $i^Y$
do not have to be related, but we have the option to define $\approx$
so that they are. In fact, the structure $\langle
\{i^X,j^X,i^Y,j^Y,i^Z,j^Z\}, \approx \rangle$ for which $\approx$
relates all (and only) identically named parameters is a control
structure.

\end{example}
Using a control structure $\langle \varset{C},\approx\rangle$,
we can ensure that all equations have the same set of CFPs. This can
be done by assigning unique names to
identical CFPs and by adding CFPs that
do not appear in an equation as formal parameters for this equation.
Without loss of generality
we therefore continue to work under the following assumption.
\begin{assumption} \label{ass:names}
The set of CFPs is the same for every equation in
a PBES; that is, for all $X, Y \in \bnd{\PBES}$ in a PBES $\PBES$ we have
$d^X \in \param{X}$ is a CFP iff $d^Y \in \param{Y}$ is a CFP, and $d^X \approx
d^Y$.

\end{assumption}
From hereon, we call any formal parameter that is not a control flow parameter
a \emph{data parameter}. We make this
distinction explicit by partitioning $\varset{D}$ into CFPs $\varset{C}$
and data parameters $\varset{D}^{\DP}$. As a consequence of Assumption~\ref{ass:names},
we may assume that every PBES we consider has equations with the same sequence of CFPs;
\ie, all equations are of the form
$\sigma X(\var{c} \colon \vec{C}, \var{d^X} \colon \vec{D^X})
= \rhs{X}(\var{c}, \var{d^X})$, where $\var{c}$ is the (vector of) CFPs, and
$\var{d^X}$ is the (vector of) data parameters of the equation for $X$. \medskip


Using the CFPs, we next construct a control flow graph. Vertices in this
graph represent valuations for the vector of CFPs and
the edges capture dependencies on PVIs.
The set of potential valuations for the CFPs is bounded by $\values{\var[k]{c}}$, 
defined as:
  \[
  \{ \init{\var[k]{c}} \} \cup \bigcup\limits_{i \in \mathbb{N}, X \in \bnd{\PBES}}
  \{ v \in D \mid 
  \source{X}{i}{k} = v \lor \dest{X}{i}{k} = v \}. 
  \]
We generalise $\valuesname$ to the vector $\var{c}$ in the obvious way.
\begin{definition}
  \label{def:globalCFGHeuristic}
The control flow graph (CFG) of $\PBES$ is a directed graph $(\VCFGsyn, 
{\smash{\ECFGsyn}})$ 
with:
\begin{compactitem}
\item $\VCFGsyn \subseteq \bnd{\PBES} \times \values{\var{c}}$.

\item  ${\smash{\ECFGsyn}} \subseteq V \times \mathbb{N} \times
V$ is the least relation for which, whenever $(X,\val{v}) \ECFGsyn[i]
(\predphi{\rhs{X}}{i},\val{w})$ then for every $k$ either:
\begin{compactitem}
\item $\source{X}{i}{k} = \val[k]{v}$ and $\dest{X}{i}{k} = \val[k]{w}$, or

\item $\source{X}{i}{k} = \bot$, $\copied{X}{i}{k} = k$ and $\val[k]{v} = 
\val[k]{w}$, or
\item $\source{X}{i}{k} = \bot$, and $\dest{X}{i}{k} = \val[k]{w}$.   

\end{compactitem}
\end{compactitem}
\end{definition}
We refer to the vertices in the CFG as \emph{locations}. Note that
a CFG is finite since the set $\values{\var{c}}$ is finite.
Furthermore, CFGs are complete in the sense that all PVIs on which
the truth of some $\rhs{X}$ may depend when $\val{c} = \val{v}$ are neighbours 
of 
location $(X, \val{v})$. 
\reportonly{
  \begin{restatable}{lemma}{resetrelevantpvineighbours}
    \label{lem:relevant_pvi_neighbours}
    Let $(\VCFGsyn, {\smash{\ECFGsyn}})$ be $\PBES$'s control flow graph. Then
    for all $(X, \val{v}) \in \VCFGsyn$ and all predicate environments
    $\predenv, \predenv'$ and data environments $\dataenv$:
    $$\sem{\rhs{X}}{\predenv}{\dataenv[\subst{\var{c}}{\sem{\val{v}}{}}]}
    = \sem{\rhs{X}}{\predenv'}{\dataenv[\subst{\var{c}}{\sem{\val{v}}{}}]}$$
    provided that
    $\predenv(Y)(\val{w}) = \predenv'(Y)(\val{w})$ for all
    $(Y, \val{w})$ satisfying
    $(X, \val{v}) \ECFGsyn[i] (Y, \val{w})$.
    
  \end{restatable}
  \begin{proof}
    Let $\predenv, \predenv'$ be predicate environments,$\dataenv$ a data
    environment, and let $(X, \val{v}) \in V$.
    Suppose that for all $(Y, \val{w})$ for which
    $(X, \val{v}) \ECFGsyn[i] (Y, \val{w})$, we know that $\predenv(Y)(\val{w}) 
    =
    \predenv'(Y)(\val{w})$.
    
    Towards a contradiction, let
    $\sem{\rhs{X}}{\predenv}{\dataenv[\subst{\var{c}}{\sem{\val{v}}{}{}}]}
    \neq 
    \sem{\rhs{X}}{\predenv'}{\dataenv[\subst{\var{c}}{\sem{\val{v}}{}{}}]}$.
    Then there must be a predicate variable instance $\predinstphi{\rhs{X}, i}$
    such that
    \begin{equation}\label{eqn:ass_pvi}
      \begin{array}{cl}
        & 
        \predenv(\predphi{\rhs{X}}{i})(\sem{\dataphi{\rhs{X}}{i}}{}{\dataenv[\subst{\var{c}}{\sem{\val{v}}{}{}}]})\\
        \neq & 
        \predenv'(\predphi{\rhs{X}}{i})(\sem{\dataphi{\rhs{X}}{i}}{}{\dataenv[\subst{\var{c}}{\sem{\val{v}}{}{}}]}).
      \end{array}
    \end{equation}
    Let
    $\dataphi{\rhs{X}}{i} = (\val{e},\val{e'})$, where $\val{e}$ are the values
    of the control flow parameters, and $\val{e'}$ are the values of the data
    parameters.
    
    Consider an arbitrary control flow parameter $\var[\ell]{c}$. We
    distinguish two cases:
    \begin{compactitem}
      \item $\source{X}{i}{\ell} \neq \bot$. Then we know
      $\dest{X}{i}{\ell} \neq \bot$, and the
      requirement for the edge  $(X, \val{v}) \ECFGsyn[i] 
      (\predphi{\rhs{X}}{i}, \val{e})$
      is satisfied for $\ell$.
      \item $\source{X}{i}{\ell} = \bot$. Since $\var[\ell]{c}$
      is a control flow parameter, we can distinguish two cases based on
      Definitions~\ref{def:LCFP} and \ref{def:GCFP}:
      \begin{compactitem}
        \item $\dest{X}{i}{\ell} \neq \bot$.  Then parameter $\ell$
        immediately satisfies the requirements that show the existence
        of the edge $(X, \val{v}) \ECFGsyn[i]
        (\predphi{\rhs{X}}{i}, \val{e})$ in the third clause in the
        definition of CFG.
        \item $\copied{X}{i}{\ell} = \ell$.
        According to the definition of $\copiedname$, we now know that
        $\val[\ell]{v} = \val[\ell]{e}$, hence the edge
        $(X, \val{v}) \ECFGsyn[i]
        (\predphi{\rhs{X}}{i}, \val{e})$ exists according to the
        second requirement in the definition of CFG.
      \end{compactitem}
    \end{compactitem}
    
    Since we have considered an arbitrary $\ell$, we know that for all $\ell$
    the requirements are satisfied, hence
    $(X, \val{v}) \ECFGsyn[i] (\predphi{\rhs{X}}{i}, \val{e})$. Then
    according to the definition of $\predenv$ and $\predenv'$,
    $\predenv(\predphi{\rhs{X}}{i})(\sem{\val{e}}{}{\dataenv[\subst{\var{c}}{\sem{\val{v}}{}{}}]})
    = 
    \predenv'(\predphi{\rhs{X}}{i})(\sem{\val{e}}{}{\dataenv[\subst{\var{c}}{\sem{\val{v}}{}{}}]})$.
    This contradicts \eqref{eqn:ass_pvi}, hence we find that 
    $\sem{\rhs{X}}{\predenv}{\dataenv[\subst{\var{c}}{\sem{\var{v}}{}{}}]}
    = 
    \sem{\rhs{X}}{\predenv'}{\dataenv[\subst{\var{c}}{\sem{\var{v}}{}{}}]}$.\qedhere
  \end{proof}
} 
\begin{example}
Using the CFPs identified earlier and an appropriate unicity constraint, we can
obtain the CFG depicted in Fig.~\ref{fig:CFG} for our running example.
\end{example}

\paragraph{Implementation.} CFGs are defined in terms of
CFPs, which in turn are obtained from a unicity constraint.
Our definition of a unicity constraint is not constructive. 
However, a unicity constraint can be derived from \emph{guards} for
a PVI. While computing the exact guard, \ie the strongest formula $\psi$ satisfying 
$\varphi \equiv
\varphi[i \mapsto (\psi \wedge \predinstphi{\varphi}{i})]$, is 
computationally 
hard, we can efficiently approximate it as follows:
\begin{definition}
Let $\varphi$ be a predicate formula. We define the \emph{guard}
of the $i$-th PVI in $\varphi$,
denoted $\guard{i}{\varphi}$, inductively as follows:
\begin{align*}
\guard{i}{b} & = \false &
\guard{i}{Y} & = \true \\
\guard{i}{\forall d \colon D . \varphi} & = \guard{i}{\varphi} &
\guard{i}{\exists d \colon D . \varphi} & = \guard{i}{\varphi} \\
\guard{i}{\varphi \land \psi} & = \begin{cases}
s(\varphi) \land \guard{i - \npred{\varphi}}{\psi} & \text{if } i > \npred{\varphi} \\
s(\psi) \land \guard{i}{\varphi} & \text{if } i \leq \npred{\varphi}
\end{cases} \span\omit\span\omit\\
\guard{i}{\varphi \lor \psi} & = \begin{cases}
s(\lnot \varphi) \land \guard{i - \npred{\varphi}}{\psi} & \text{if } i > 
\npred{\varphi} \\
s(\lnot \psi) \land \guard{i}{\varphi} & \text{if } i \leq \npred{\varphi}
\end{cases}\span\omit\span\omit
\end{align*}
where
$s(\varphi) = \varphi$ if $\npred{\varphi} = 0$, and $\true$ otherwise.
\end{definition}
We have $\varphi \equiv \varphi[i \mapsto (\guard{i}{\varphi} \wedge \predinstphi{\varphi}{i})]$;
\ie, 
$\predinstphi{\varphi}{i}$ is relevant to $\varphi$'s truth value only if
$\guard{i}{\varphi}$ is satisfiable.
\reportonly{
This is formalised int he following lemma.
\begin{lemma}\label{lem:addGuard}
Let $\varphi$ be a predicate formula, and let $i \leq \npred{\varphi}$, then
for every predicate environment $\predenv$ and data environment $\dataenv$,
$$
\sem{\varphi}{\predenv}{\dataenv}
  = \sem{\varphi[i \mapsto (\guard{i}{\varphi} \land 
  \predinstphi{\varphi}{i})]}{\predenv}{\dataenv}. 
$$
\end{lemma}
\begin{proof}
Let $\predenv$ and $\dataenv$ be arbitrary. We proceed by induction on 
$\varphi$.
The base cases where $\varphi = b$ and $\varphi = Y(\val{e})$ are trivial, and
$\forall d \colon D . \psi$ and $\exists d \colon D . \psi$ follow immediately
from the induction hypothesis. We
describe the case where $\varphi = \varphi_1 \land \varphi_2$ in detail,
the $\varphi = \varphi_1 \lor \varphi_2$ is completely analogous.

Assume that $\varphi = \varphi_1 \land \varphi_2$.
Let $i \leq \npred{\varphi_1 \land \varphi_2}$.
 Without loss of generality assume that $i \leq \npred{\varphi_1}$, the other
 case is analogous. According to the induction hypothesis,
 \begin{equation}
 \sem{\varphi_1}{\predenv}{\dataenv}
 = \sem{\varphi_1[i \mapsto (\guard{i}{\varphi_1} \land 
 \predinstphi{\varphi_1}{i})]}{\predenv}{\dataenv} \label{eq:ih} 
 \end{equation}
 We distinguish two cases.
 \begin{compactitem}
   \item $\npred{\varphi_2} \neq 0$. Then 
   $\sem{\guard{i}{\varphi_1}}{\dataenv}{\predenv}
   = \sem{\guard{i}{\varphi_1 \land \varphi_2}}{\dataenv}{\predenv}$
 according to the definition of $\guardname$. Since $i \leq \npred{\varphi_1}$,
 we find that
 $
 \sem{\varphi_1 \land \varphi_2}{\predenv}{\dataenv}
 = \sem{(\varphi_1 \land \varphi_2)[i \mapsto (\guard{i}{\varphi_1 \land 
 \varphi_2} \land \predinstphi{\varphi_1 \land 
 \varphi_2}{i})]}{\predenv}{\dataenv}. 
 $
   \item $\npred{\varphi_2} = 0$.
   We have to show that
   $$\sem{\varphi_1 \land \varphi_2}{\predenv}{\dataenv}
   = \sem{\varphi_1[i \mapsto (\guard{i}{\varphi_1 \land \varphi_2} \land 
   \predinstphi{\varphi_1}{i})] \land \varphi_2}{\predenv}{\dataenv}$$
   
   From the semantics, it follows that
   $\sem{\varphi_1 \land \varphi_2}{\predenv}{\dataenv}
   = \sem{\varphi_1}{\predenv}{\dataenv} \land 
   \sem{\varphi_2}{\predenv}{\dataenv}.
   $
   Combined with \eqref{eq:ih}, and an application of the semantics, this yields
   $$
   \sem{\varphi_1 \land \varphi_2}{\predenv}{\dataenv}
   = \sem{\varphi_1[i \mapsto (\guard{i}{\varphi_1} \land 
   \predinstphi{\varphi_1}{i})] \land \varphi_2}{\predenv}{\dataenv}.
   $$
   According to the definition
   of $\guardname$, $\guard{i}{\varphi_1 \land \varphi_2} = \varphi_2 \land 
   \guard{i}{\varphi_1}$.
   Since $\varphi_2$ is present in the context, the desired result 
   follows.\qedhere
 \end{compactitem}
\end{proof}
We can generalise the above, and guard every predicate variable
instance in a formula with its guard, which preserves the solution of the
formula. To this end we introduce the function $\guardedname$.
\begin{definition}\label{def:guarded}
Let $\varphi$ be a predicate formula, then
$$\guarded{\varphi} \isdef \varphi[i \mapsto (\guard{i}{\varphi} \land 
\predinstphi{\varphi}{i})]_{i \leq \npred{\varphi}}$$
where $[i \mapsto \psi_i]_{i \leq \npred{\varphi}}$ is the simultaneous
syntactic substitution of all $\predinstphi{\varphi}{i}$ with $\psi_i$.
\end{definition}
The following corollary follows immediately from Lemma~\ref{lem:addGuard}.
\begin{corollary}\label{cor:guardedPreservesSol}
For all formulae $\varphi$, and for all predicate environments $\predenv$,
and data environments $\dataenv$,
$
\sem{\varphi}{\predenv}{\dataenv} = \sem{\guarded{\varphi}}{\predenv}{\dataenv} 
$
\end{corollary}
This corollary confirms our intuition that indeed the guards we compute
effectively guard the recursions in a formula.
 
} 
A good heuristic for defining the unicity constraints is looking
for positive occurrences of constraints of the form $d
= e$ in the guards and using this information to see if the arguments
of PVIs reduce to constants.

\section{Data Flow Analysis}\label{sec:dataflow}

Our liveness analysis is built on top of CFGs constructed using
Def.~\ref{def:globalCFGHeuristic}. The analysis proceeds as follows:
for each location in the CFG, we first identify the data parameters
that may directly affect the truth value of the corresponding predicate
formula. Then we inductively identify data parameters that can
affect such parameters through PVIs as live as well.  Upon termination,
each location is labelled by the \emph{live} parameters at that
location.
The set $\significant{\varphi}$ of parameters that affect the truth value of
a predicate formula $\varphi$, \ie, those parameters that occur in 
Boolean data terms, are approximated as follows:
\begin{align*}
\significant{b} & = \free{b} &
\significant{Y(e)} & = \emptyset \\
\significant{\varphi \land \psi} & = \significant{\varphi} \cup \significant{\psi} &
\significant{\varphi \lor \psi} & = \significant{\varphi} \cup \significant{\psi} \\
\significant{\exists d \colon D . \varphi} & = \significant{\varphi} \setminus \{ d \} &
\significant{\forall d \colon D . \varphi} & = \significant{\varphi} \setminus \{ d \}
\end{align*}
Observe that $\significant{\varphi}$ is not invariant under logical
equivalence.  We use this fact to our advantage: we assume the
existence of a function $\simplifyname$ for which we require
$\simplify{\varphi} \equiv \varphi$, and
$\significant{\simplify{\varphi}}\subseteq \significant{\varphi}$. 
An appropriately chosen function $\simplifyname$ may help to narrow
down the parameters that affect the truth value of predicate
formulae in our base case. Labelling the CFG with live variables
is achieved as follows:
\begin{definition}
\label{def:markingHeuristic}
Let $\PBES$ be a PBES and let
$(\VCFGsyn, \ECFGsyn)$ be its CFG.
The labelling $\markingsn{\syntactic}{} \colon \VCFGsyn \to
\mathbb{P}({\varset{D}^{\DP}})$ is defined as 
$\marking{\markingsn{\syntactic}{}}{X, \val{v}}  = \bigcup_{n \in \N}
\marking{\markingsn{\syntactic}{n}}{X, \val{v}}$, with
$\markingsn{\syntactic}{n}$ inductively defined as:
\[
\begin{array}{ll}
\marking{\markingsn{\syntactic}{0}}{X, \val{v}} & = 
\significant{\simplify{\rhs{X}[\var{c} := \val{v}]}}\\
\marking{\markingsn{\syntactic}{n+1}}{X, \val{v}} & =
\marking{\markingsn{\syntactic}{n}}{X, \val{v}} \\
                    & \cup \{ d \in \param{X} \cap \varset{D}^{\DP} \mid \exists {i \in \mathbb{N},
(Y,\val{w}) \in V}:
                       (X, \val{v}) \ECFGsyn[i] (Y, \val{w}) \\
                    & \qquad \land \exists {\var[\ell]{d} \in 
                    \marking{\markingsn{\syntactic}{n}}{Y, \val{w}}}:~
                                 \affects{d}{\dataphi[\ell]{\rhs{X}}{i}}
                                 \} 
\end{array}
\]

\end{definition}
The set $\marking{\markingsn{\syntactic}{}}{X, \val{v}}$ approximates the set of
parameters potentially live at location $(X, \val{v})$; all other data
parameters are guaranteed to be ``dead'', \ie, irrelevant.
\begin{example}\label{exa:globalCFGLabelled} The
labelling computed for our running example is depicted in Fig.~\ref{fig:CFG}.
One can cheaply establish that
$k^Z \notin \marking{\markingsn{\syntactic}{0}}{Z,1,2}$ since assigning
value $2$ to $j^Z$ in $Z$'s right-hand side effectively allows to
reduce subformula $(k < 10 \vee j =2)$ to $\true$. We have
$l \in \marking{\markingsn{\syntactic}{1}}{Z,1,2}$ since
we have $k^Y \in \marking{\markingsn{\syntactic}{0}}{Y,1,1}$.

\end{example}
\reportonly{
The labelling from Definition~\ref{def:markingHeuristic} induces
a relation $\markrel{\markingsn{\syntactic}{}}$ on signatures as follows.
\begin{definition}\label{def:markrelSyntactic}
Let $\markingsn{\syntactic}{} \colon V \to \mathbb{P}(\varset{D}^{\DP})$ be a 
labelling. $\markingsn{\syntactic}{}$ induces a relation
$\markrel{\markingsn{\syntactic}{}}$ such that $(X, \sem{\val{v}}{}{}, 
\sem{\val{w}}{}{})
\markrel{\markingsn{\syntactic}{}} (Y, \sem{\val{v'}}{}{}, \sem{\val{w'}}{}{})$ 
if and only if
$X = Y$, $\sem{\val{v}}{}{} = \sem{\val{v'}}{}{}$, and
$\forall \var[k]{d} \in \marking{\markingsn{\syntactic}{}}{X,
  \val{v}}: \semval[k]{w} = \semval[k]{w'}$.
\end{definition}
Observe that the relation $\markrel{\markingsn{\syntactic}{}}$
allows for relating \emph{all} instances of the non-labelled data parameters at
a given control flow location.
We prove that, if locations are related using the relation
$\markrel{\markingsn{\syntactic}{}}$, then the corresponding instances in the
PBES have the same solution by showing that $\markrel{\markingsn{\syntactic}{}}$
is a consistent correlation.

In order to prove this, we first show that given a predicate environment and two
data environments, if the solution of a formula differs between those
environments, and all predicate variable instances in the formula have the same
solution, then there must be a significant parameter $d$ in
the formula that gets a different value in the two data environments.
\begin{lemma}
\label{lem:non_recursive_free}
For all formulae $\varphi$, predicate environments $\predenv$,
and data environments $\dataenv, \dataenv'$, if 
$\sem{\varphi}{\predenv}{\dataenv} \neq \sem{\varphi}{\predenv}{\dataenv'}$
and for all $i \leq \npred{\varphi}$, 
$\sem{\predinstphi{\varphi}{i}}{\predenv}{\dataenv}
 = \sem{\predinstphi{\varphi}{i}}{\predenv}{\dataenv'}$,
then $\exists d \in \significant{\varphi}: \dataenv(d) \neq \dataenv'(d)$.
\end{lemma}
\begin{proof}
We proceed by induction on $\varphi$.
\begin{compactitem}
  \item $\varphi = b$. Trivial.
  \item $\varphi = Y(e)$. In this case the two preconditions
        contradict, and the result trivially follows.
  \item $\varphi = \forall e \colon D . \psi$. Assume that
  $\sem{\forall e \colon D . \psi}{\predenv}{\dataenv}
    \neq \sem{\forall e \colon D . \psi}{\predenv}{\dataenv'}$, and furthermore,
    $\forall i \leq \npred{\forall e \colon D . \psi}:
      \sem{\predinstphi{\forall e \colon D . \psi}{i}}{\predenv}{\dataenv}
    = \sem{\predinstphi{\forall e \colon D . \psi}{i}}{\predenv}{\dataenv'}$.
    
  According to the semantics, we have
  $\forall u \in \semset{D} . \sem{\psi}{\predenv}{\dataenv[\subst{e}{u}]}
    \neq \forall u' \in \semset{D} . 
    \sem{\psi}{\predenv}{\dataenv'[\subst{e}{u'}]}$,
  so $\exists u \in \semset{D}$ such that
  $\sem{\psi}{\predenv}{\dataenv[\subst{e}{u}]}
    \neq \sem{\psi}{\predenv}{\dataenv'[\subst{e}{u}]}$.
  Choose an arbitrary such $u$. Observe that also
  for all $i \leq \npred{\psi}$, we know that
  $$\sem{\predinstphi{\psi}{i}}{\predenv}{\dataenv[\subst{e}{u}]}
    = \sem{\predinstphi{\psi}{i}}{\predenv}{\dataenv'[\subst{e}{u}]}.$$
  According to the induction hypothesis, there exists some $d \in 
  \significant{\psi}$
  such that $\dataenv[\subst{e}{u}](d) \neq \dataenv'[\subst{e}{u}](d)$.
  Choose such a $d$, and observe that $d \neq e$ since otherwise $u \neq u$,
  hence $d \in \significant{\forall e \colon D . \psi}$,
  which is the desired result.
  \item $\varphi = \exists e \colon D . \psi$. Analogous to the previous case.
  \item $\varphi = \varphi_1 \land \varphi_2$. Assume that
  $\sem{\varphi_1 \land \varphi_2}{\predenv}{\dataenv}
    \neq \sem{\varphi_1 \land \varphi_2}{\predenv}{\dataenv'}$, and suppose
    that that for all $i \leq \npred{\varphi_1 \land \varphi_2}$, we know that
      $\sem{\predinstphi{\varphi_1 \land \varphi_2}{i}}{\predenv}{\dataenv}
    = \sem{\predinstphi{\varphi_1 \land \varphi_2}{i}}{\predenv}{\dataenv'}$.
  According to the first assumption, either
  $\sem{\varphi_1}{\predenv}{\dataenv} \neq 
  \sem{\varphi_1}{\predenv}{\dataenv'}$,
  or $\sem{\varphi_2}{\predenv}{\dataenv} \neq 
  \sem{\varphi_2}{\predenv}{\dataenv'}$.
  
  Without loss of generality, assume that
  $\sem{\varphi_1}{\predenv}{\dataenv} \neq 
  \sem{\varphi_1}{\predenv}{\dataenv'}$,
  the other case is completely analogous.
  Observe that from our second assumption it follows that
  $\forall i \leq \npred{\varphi_1}:
      \sem{\predinstphi{\varphi_1}{i}}{\predenv}{\dataenv}
    = \sem{\predinstphi{\varphi_1}{i}}{\predenv}{\dataenv'}$.
  According to the induction hypothesis, we now find some
  $d \in \significant{\varphi_1}$ such that $\dataenv(d) \neq \dataenv'(d)$.
  Since   $\significant{\varphi_1} \subseteq \significant{\varphi_1 \land 
  \varphi_2}$,
  our result follows.
  \item $\varphi = \varphi_1 \lor \varphi_2$. Analogous to the previous case. 
  \qedhere
\end{compactitem}
\end{proof}
This is now used in proving the following proposition, that shows that related
signatures have the same solution. This result follows from the fact that
$\markrel{\markingsn{\syntactic}{}}$ is a consistent correlation.
\begin{proposition}
\label{prop:ccSyn}
Let $\PBES$ be a PBES, with global control flow graph $(\VCFGsyn, \ECFGsyn)$,
and labelling $\markingsn{\syntactic}{}$. For all predicate environments
$\predenv$ and data environments $\dataenv$,
$$(X, \semval{v}, \semval{w}) \markrel{\markingsn{\syntactic}{}}
(Y, \semval{v'}, \semval{w'})
\implies \sem{\PBES}{\predenv}{\dataenv}(X(\val{v},\val{w})) =
\sem{\PBES}{\predenv}{\dataenv}(Y(\val{v'},\val{w'})).$$
\end{proposition}
\begin{proof}
We show that $\markrel{\markingsn{\syntactic}{}}$ is a consistent correlation.
The result then follows immediately from Theorem~\ref{thm:cc}.

Let $n$ be the smallest number such that for all $X, \val{v}$,
$\markingsn{\syntactic}{n+1}(X, \val{v})
= \markingsn{\syntactic}{n}(X, \val{v})$, and hence
$\markingsn{\syntactic}{n}(X, \val{v}) = \markingsn{\syntactic}{}(X, \val{v})$.
Towards a contradiction, suppose that $\markrel{\markingsn{\syntactic}{}}$
is not a consistent correlation. Since $\markrel{\markingsn{\syntactic}{}}$
is not a consistent correlation, there exist
$X, X', \val{v}, \val{v'}, \val{w}, \val{w'}$ such that
$(X, \semval{v}, \semval{w}) \markrel{\markingsn{\syntactic}{n}} (X', 
\semval{v'}, \semval{w'})$,
and
\begin{equation*}
\exists \predenv \in \correnv{\markrel{\markingsn{\syntactic}{n}}}, \dataenv:
\sem{\rhs{X}}{\predenv}{\dataenv[\subst{\var{c}}{\semval{v}}, 
\subst{\var{d}}{\semval{w}}]}
\neq \sem{\rhs{X'}}{\predenv}{\dataenv[\subst{\var{c}}{\semval{v'}}, 
\subst{\var{d}}{\semval{w'}}]}.
\end{equation*}
According to Definition~\ref{def:markrelSyntactic}, $X = X'$, and $\semval{v} = 
\semval{v'}$,
hence this is equivalent to
\begin{equation}\label{eq:equal_phi}
\exists \predenv \in \correnv{\markrel{\markingsn{\syntactic}{n}}}, \dataenv: 
\sem{\rhs{X}}{\predenv}{\dataenv[\subst{\var{c}}{\semval{v}}, 
\subst{\var{d}}{\semval{w}}]}
\neq \sem{\rhs{X}}{\predenv}{\dataenv[\subst{\var{c}}{\semval{v}}, 
\subst{\var{d}}{\semval{w'}}]}.
\end{equation}
Let $\predenv$ and $\dataenv$ be such, and let
$\dataenv_1 = \dataenv[\subst{\var{c}}{\semval{v}}, 
\subst{\var{d}}{\semval{w}}]$
and $\dataenv_2 = \dataenv[\subst{\var{c}}{\semval{v}}, 
\subst{\var{d}}{\semval{w'}}]$.
Define $ \varphi'_X \isdef \simplify{\rhs{X}[\var{c} := \val{v}]}.$
Since the values in $\val{v}$ are closed, and from the definition of 
$\simplifyname$,
we find that $\sem{\rhs{X}}{\predenv}{\dataenv_1} = 
\sem{\varphi'_X}{\predenv}{\dataenv_1}$,
and likewise for $\dataenv_2$. Therefore, we know that
\begin{equation}
\label{eq:equal_phi'}
\sem{\varphi'_X}{\predenv}{\dataenv_1} \neq 
\sem{\varphi'_X}{\predenv}{\dataenv_2}.
\end{equation}

Observe that for all $\var[k]{d} \in \marking{\markingsn{\syntactic}{}}{X, 
\val{v}}$, 
$\semval[k]{w} = \semval[k]{w'}$ by definition of
$\markrel{\markingsn{\syntactic}{}}$.
Every predicate variable instance that might change the solution of $\varphi'_X$
is a neighbour of $(X, \val{v})$ in the control flow graph, according to Lemma
\ref{lem:relevant_pvi_neighbours}.
Take an arbitrary predicate variable instance
$\predinstphi{\rhs{X}}{i} = Y(\val{e}, \val{e'})$ in $\varphi'_X$.
We first show that $\sem{\val[\ell]{e'}}{}{\dataenv_1}
  = \sem{\val[\ell]{e'}}{}{\dataenv_2}$ for all $\ell$.

Observe that $\sem{\val{e}}{}{\dataenv_1} = \sem{\val{e}}{}{\dataenv_2}$ since
$\val{e}$ are expressions substituted for control flow parameters, and hence
are either constants, or the result of copying.

Furthermore, there is no unlabelled parameter $\var[k]{d}$ that can influence a 
labelled parameter $\var[\ell]{d}$ at location $(Y, \val{u})$. If there is a
$\var[\ell]{d} \in \markingsn{\syntactic}{n}(Y, \val{u})$ such that
$\var[k]{d} \in \free{\val[\ell]{e'}}$, and
$\var[k]{d} \not \in \markingsn{\syntactic}{n}(X, \val{v})$, then by
definition of labelling $\var[k]{d} \in \markingsn{\syntactic}{n+1}(X, 
\val{v})$,
which contradicts the assumption that the labelling is stable, so it follows 
that
\begin{equation} \label{eq:equivalent_arguments_Xi}
\sem{\val[\ell]{e'}}{}{\dataenv_1}
  = \sem{\val[\ell]{e'}}{}{\dataenv_2}\text{ for all }\ell.
\end{equation} 
  
From \eqref{eq:equivalent_arguments_Xi}, and since we have chosen the predicate
variable instance arbitrarily, it follows that for all $1 \leq i \leq 
\npred{\varphi'_X}$,
$\sem{X(\val{e},\val{e'})}{\predenv}{\dataenv_1}
  = \sem{X(\val{e},\val{e'})}{\predenv}{\dataenv_2}$.
Together with \eqref{eq:equal_phi'}, according to 
Lemma~\ref{lem:non_recursive_free},
this implies that there is some $d \in \significant{\varphi'_X}$
such that $\dataenv_1(d) \neq \dataenv_2(d)$. From the definition of
$\markingsn{\syntactic}{0}$, however, it follows that $d$ must be labelled
in $\marking{\markingsn{\syntactic}{0}}{X, \val{v}}$, and hence also in 
$\marking{\markingsn{\syntactic}{n}}{X, \val{v}}$.
According to the definition of $\markrel{\markingsn{\syntactic}{n}}$ it then is
the case that $\dataenv_1(d) = \dataenv_2(d)$, which is a contradiction.
Since also in this case we derive a contradiction, the original assumption that
$\markrel{\markingsn{\syntactic}{}}$ is not a consistent correlation does not
hold, and we conclude that $\markrel{\markingsn{\syntactic}{}}$ is a consistent
correlation.  \qedhere
\end{proof}
} 
%
\label{sec:reset}
A parameter $d$ that is not live at a location 
can be assigned a fixed default value. To this end
the corresponding data argument of the PVIs that lead to that location
are replaced by a default value $\init{d}$. This is achieved by
function $\resetname$, defined below:
\reportonly{
\begin{definition}
\label{def:reset}
Let $\PBES$ be a PBES, let $(V, \to)$ be its CFG, with labelling
$\markingsn{\syntactic}{}$\!. Resetting a PBES is inductively defined on the 
structure of
$\PBES$.
$$
\begin{array}{lcl}
\reset{\markingsn{\syntactic}{}}{\emptyPBES} & \isdef & \emptyPBES \\
\reset{\markingsn{\syntactic}{}}{\sigma X(\var{c} \colon \vec{C}, \var{d} 
\colon \vec{D}) = \varphi) \PBES'}  & \isdef
         & (\sigma \changed{X}(\var{c} \colon \vec{C}, \var{d} \colon\vec{D}) = 
         \reset{\markingsn{\syntactic}{}}{\varphi}) 
         \reset{\markingsn{\syntactic}{}}{\PBES'} \\

\end{array}
$$
Resetting for formulae is defined inductively as follows:
$$
\begin{array}{lcl}
\reset{\markingsn{\syntactic}{}}{b} & \isdef & b\\
\reset{\markingsn{\syntactic}{}}{\varphi \land \psi} & \isdef & 
\reset{\markingsn{\syntactic}{}}{\varphi} \land 
\reset{\markingsn{\syntactic}{}}{\psi}\\
\reset{\markingsn{\syntactic}{}}{\varphi \lor \psi} & \isdef & 
\reset{\markingsn{\syntactic}{}}{\varphi} \lor 
\reset{\markingsn{\syntactic}{}}{\psi}\\
\reset{\markingsn{\syntactic}{}}{\forall d \colon D. \varphi} & \isdef & 
\forall d \colon D. \reset{\markingsn{\syntactic}{}}{\varphi} \\
\reset{\markingsn{\syntactic}{}}{\exists d \colon D. \varphi} & \isdef & 
\exists d \colon D. \reset{\markingsn{\syntactic}{}}{\varphi} \\
\reset{\markingsn{\syntactic}{}}{X(\val{e}, \val{e'})} & \isdef &
   \bigwedge_{\val{v} \in \values{\var{c}}} (\val{e} = \val{v} \implies 
   \changed{X}(\val{v}, \resetvars{(X, 
   \val{v})}{\markingsn{\syntactic}{}}{\val{e'}}))
\end{array}
$$
With $\val{e} = \val{v}$ we denote that for all $i$,
$\val[i]{e} = \var[i]{v}$.
The function $\resetvars{(X, \val{v})}{\markingsn{\syntactic}{}}{\val{e'}}$ is 
defined 
positionally as follows:
$$
\ind{\resetvars{(X, \val{v})}{\markingsn{\syntactic}{}}{\val{e'}}}{i} = 
\begin{cases}
\val[i]{e'} & \text{ if } \var[i]{d} \in
\marking{\markingsn{}{}}{X, \val{v}}
\\
\init{\var[i]{d}} & \text{otherwise}.
\end{cases}
$$
\end{definition}

\begin{remark}
We can reduce the number of equivalences we introduce in resetting a recurrence.
This effectively reduces the guard as follows.

Let $X \in \bnd{\PBES}$, such that $Y(\val{e}, \val{e'}) = 
\predinstphi{\rhs{X}}{i}$,
and let $I = \{ j \mid \dest{X}{i}{j} = \bot \} $ denote the indices of the
control flow parameters for which the destination is undefined.

Define $\var{c'} = \var[i_1]{c}, \ldots, \var[i_n]{c}$ for $i_n \in I$,
and $\val{f} = \val[i_1]{e}, \ldots, \val[i_n]{e}$ to be the vectors
of control flow parameters for which the destination is undefined, and the
values that are assigned to them in predicate variable instance $i$. Observe
that these are the only control flow parameters that we need to constrain in
the guard while resetting.

We can redefine $\reset{\markingsn{\syntactic}{}}{X(\val{e}, \val{e'})}$ as
follows.
$$\reset{\markingsn{\syntactic}{}}{X(\val{e}, \val{e'})} \isdef
   \bigwedge_{\val{v'} \in \values{\var{c'}}}
        ( \val{f} = \val{v'}  \implies 
            \changed{X}(\val{v}, \resetvars{\markingsn{\syntactic}{}}{(X, 
            \val{v})}{\val{e'}}) ).$$
In this definition
$\val{v}$ is defined positionally as $$\val[j]{v} = 
 \begin{cases}
 \val[j]{v'} & \text{if } j \in I \\
 \dest{X}{i}{j} & \text{otherwise}
 \end{cases}$$
\end{remark}
Resetting dead parameters preserves the solution of the PBES. We formalise
this in Theorem~\ref{thm:resetSound} below. Our proof is based on consistent
correlations. We first define the relation $R^{\resetname}$, and we show that
this is indeed a consistent correlation. Soundness then follows from
Theorem~\ref{thm:cc}. Note that  $R^{\resetname}$ uses the relation
$\markrel{\markingsn{\syntactic}{}}$ from Definition~\ref{def:markrelSyntactic}
to relate predicate variable instances of the original equation system. The 
latter
is used in the proof of Lemma~\ref{lem:resetRecursion}.
\begin{definition}\label{def:resetRelSyn}
Let $R^{\resetname}$ be the relation defined as follows.
$$
\begin{cases}
X(\semval{v}, \semval{w})) R^{\resetname} \changed{X}(\semval{v},
\sem{\resetvars{(X, \val{v})}{\markingsn{\syntactic}{}}{\val{w})}}{}{}) \\
X(\semval{v}, \semval{w}) R^{\resetname} X(\semval{v}, \semval{w'}) & \text{if }
 X(\semval{v}, \sem{\val{w}}{}{}) \markrel{\markingsn{\syntactic}{}} 
 X(\semval{v}, \semval{w'})
\end{cases}
$$
\end{definition}
We first show
that we can unfold the values of the control flow parameters in every predicate
variable instance, by duplicating the predicate variable instance, and
substituting the values of the CFPs.
\begin{lemma}\label{lem:unfoldCfl}
Let $\predenv$ and $\dataenv$ be environments, and let $X \in \bnd{\PBES}$,
then for all $i \leq \npred{\rhs{X}}$, such that $\predinstphi{\rhs{X}}{i}
= Y(\val{e},\val{e'})$,
$$\sem{Y(\val{e}, \val{e'}))}{\predenv}{\dataenv} =
\sem{\bigwedge_{\val{v}\in \values{\var{c}}} (\val{e} = \val{v} \implies
Y(\val{v}, \val{e'})}{\predenv}{\dataenv}$$
\end{lemma}
\begin{proof}
Straightforward; observe that $\val{e} = \val{v}$ for
exactly one $\val{v} \in \values{\var{c}}$, using that $\val{v}$ is 
closed.\qedhere
\end{proof}
Next we establish that resetting dead parameters is sound, \ie it
preserves the solution of the PBES. We first show that resetting
a predicate variable instance in an $R^{\resetname}$-correlating environment
and a given data environment is sound.
\begin{lemma}\label{lem:resetRecursion}
Let $\PBES$ be a PBES, let $(V, \to)$ be its CFG, with labelling
$\markingsn{\syntactic}{}$ such that $\markrel{\markingsn{\syntactic}{}}$ is a
consistent correlation, then
$$\forall \predenv \in
\correnv{R^{\resetname}}, \dataenv: \sem{Y(\val{e}, 
\val{e'})}{\predenv}{\dataenv} =
\sem{\reset{\markingsn{\syntactic}{}}{Y(\val{e}, 
\val{e'})}}{\predenv}{\dataenv}$$
\end{lemma}
\begin{proof}
Let $\predenv \in \correnv{R^{\resetname}}$, and $\dataenv$ be arbitrary. We
derive this as follows.
  $$
  \begin{array}{ll}
    & \sem{\reset{\markingsn{\syntactic}{}}{Y(\val{e}, 
    \val{e'}))}}{\predenv}{\dataenv} \\
  = & \{ \text{Definition~\ref{def:reset}} \} \\
    & \sem{\bigwedge_{\val{v} \in \cfl{Y}} (\val{e} = \val{v}
  \implies \changed{Y}(\val{v},\resetvars{(Y,
                     \val{v})}{\markingsn{\syntactic}{}}{\val{e'}}))}{\predenv}{\dataenv}
                     \\
  =^{\dagger} & \bigwedge_{\val{v} \in \cfl{Y}} (\sem{\val{e}}{}{\dataenv}
  = \sem{\val{v}}{}{} \implies
  \sem{\changed{Y}(\val{v},\resetvars{(Y,
  \val{v})}{\markingsn{\syntactic}{}}{\val{e'}}}{\predenv}{\dataenv}))
  \\
  =^{\dagger} & \bigwedge_{\val{v} \in \cfl{Y}} (\sem{\val{e}}{}{\dataenv}
  = \semval{v} \implies
  \predenv(\changed{Y})(\sem{\val{v}}{}{\dataenv},\sem{\resetvars{(Y,
  \val{v})}{\markingsn{\syntactic}{}}{\val{e'}}}{}{\dataenv}))\\
  = & \{ \predenv \in \correnv{R^{\resetname}} \} \\
    & \bigwedge_{\val{v} \in \cfl{Y}} (\sem{\val{e}}{}{\dataenv} =
  \semval{v} \implies
  \predenv(Y)(\sem{\val{v}}{}{\dataenv},\sem{\val{e'}}{}{\dataenv})))\\
  =^{\dagger} & \bigwedge_{\val{v} \in \cfl{Y}} (\sem{\val{e}}{}{\dataenv}
  = \semval{v} \implies
  \sem{Y(\val{v}, \val{e'})}{\predenv}{\dataenv})) \\
  =^{\dagger} & \sem{\bigwedge_{\val{v} \in \cfl{Y}} (\val{e} = \val{v}
  \implies Y(\val{v}, \val{e'}))}{\predenv}{\dataenv} \\
  = & \{ \text{Lemma~\ref{lem:unfoldCfl}} \}\\
    & \sem{Y(\val{e}, \val{e'}))}{\predenv}{\dataenv}
  \end{array}
  $$
  Here at $^{\dagger}$ we have used the semantics.\qedhere
\end{proof}
By extending this result to the right-hand sides of equations, we can prove that
$R^{\resetname}$ is a consistent correlation. 
\begin{proposition}
\label{prop:resetCc}
Let $\PBES$ be a PBES, and let $(V, \to)$ be a CFG, with labelling
$\markingsn{\syntactic}{}$ such that $\markrel{\markingsn{\syntactic}{}}$ is a
consistent correlation. Let $X \in \bnd{\PBES}$, with $\val{v} \in \CFL(X)$,
then for all $\val{w}$, and for all predicate environments $\predenv \in 
\correnv{R^{\resetname}}$ and data environments $\dataenv$
$$
\sem{\rhs{X}}{\predenv}{\dataenv[\subst{\var{c}}{\semval{v}},\subst{\var{d}}{\semval{w}}]}
=
\sem{\reset{\markingsn{\syntactic}{}}{\rhs{X}}}{\predenv}{\dataenv[\subst{\var{c}}{\semval{v}},\subst{\var{d}}{\sem{\resetvars{(X,
\var{v})}{\markingsn{\syntactic}{}}{\val{w}}}{}{}}]} $$
\end{proposition}
\begin{proof}
Let $\predenv$ and $\dataenv$ be arbitrary, and define $\dataenv_r \isdef
\dataenv[\subst{\var{c}}{\semval{v}},\subst{\var{d}}{\sem{\resetvars{(X,
\val{v})}{\markingsn{\syntactic}{}}{\val{w}}}{}{}}]$.
We first prove that
\begin{equation}
\sem{\rhs{X}}{\predenv}{\dataenv_r}
=
\sem{\reset{\markingsn{\syntactic}{}}{\rhs{X}}}{\predenv}{\dataenv_r}
\end{equation}
We proceed by induction on $\rhs{X}$.
\begin{compactitem}
  \item $\rhs{X} = b$. Since $\reset{\markingsn{\syntactic}{}}{b} = b$ this
        follows immediately.
  \item $\rhs{X} = Y(\val{e})$. This follows immediately from
        Lemma~\ref{lem:resetRecursion}.
  \item $\rhs{X} = \forall y \colon D . \varphi$. We derive that
        $\sem{\forall y \colon D . \varphi}{\predenv}{\dataenv_r} = \forall v 
        \in \semset{D} . \sem{\varphi}{\predenv}{\dataenv_r[\subst{y}{v}]}$.
        According to the induction hypothesis, and since we applied only a
        dummy transformation on $y$, we find that
        $\sem{\varphi}{\predenv}{\dataenv_r[\subst{y}{v}]}
        = 
        \sem{\reset{\markingsn{\syntactic}{}}{\varphi}}{\predenv}{\dataenv_r[\subst{y}{v}]}$,
        hence $\sem{\forall y \colon D . \varphi}{\predenv}{\dataenv_r} =
        \sem{\reset{\markingsn{\syntactic}{}}{\forall y \colon D . 
        \varphi}}{\predenv}{\dataenv_r}$.
  \item $\rhs{X} = \exists y \colon D . \varphi$. Analogous to the previous 
  case.
  \item $\rhs{X} = \varphi_1 \land \varphi_2$. We derive that
        $\sem{\varphi_1 \land \varphi_2}{\predenv}{\dataenv_r} =
        \sem{\varphi_1}{\predenv}{\dataenv_r} \land 
        \sem{\varphi_2}{\predenv}{\dataenv_r}$.
        If we apply the induction hypothesis on both sides we get
         $\sem{\varphi_1 \land \varphi_2}{\predenv}{\dataenv_r} =
        \sem{\reset{\markingsn{\syntactic}{}}{\varphi_1}}{\predenv}{\dataenv_r} 
        \land 
        \sem{\reset{\markingsn{\syntactic}{}}{\varphi_2}}{\predenv}{\dataenv_r}$.
        Applying the semantics, and the definition of $\resetname$ we find this 
        is equal to
        $\sem{\reset{\markingsn{\syntactic}{}}{\varphi_1 \land 
        \varphi_2}}{\predenv}{\dataenv_r}$.
  \item $\rhs{X} = \varphi_1 \lor \varphi_2$. Analogous to the previous case.
\end{compactitem}
Hence we find that
$\sem{\reset{\markingsn{\syntactic}{}}{\rhs{X}}}{\predenv}{\dataenv_{r}}
= \sem{\rhs{X}}{\predenv}{\dataenv_{r}}$.
It now follows immediately from the observation that
$\markrel{\markingsn{\syntactic}{}}$ is a consistent correlation, and
Definition~\ref{def:reset}, that
$\sem{\rhs{X}}{\predenv}{\dataenv_{r}} =
\sem{\rhs{X}}{\predenv}{\dataenv[\subst{\var{c}}{\semval{v}},\subst{\var{d}}{\semval{w}}]}$.
Our result follows by transitivity of $=$. \qedhere
\end{proof}
The theory of consistent correlations now gives an immediate proof of soundness
of resetting dead parameters, which is formalised by the following
theorem.
\begin{theorem}
\label{thm:resetSound}
Let $\PBES$ be a PBES, with control flow graph $(V, \to)$ and labelling
$\markingsn{}{}$\!. For all $X$, $\val{v}$ and $\val{w}$:
$$\sem{\PBES}{}{}(X(\sem{\val{v}},\sem{\val{w}})) =
\sem{\reset{\markingsn{\syntactic}{}}{\PBES}}{}{}(\changed{X}(\sem{\val{v}},\sem{\val{w}})).$$
\end{theorem}
\begin{proof}
Relation $R^{\resetname}$ is a consistent correlation, as witnessed by
Proposition~\ref{prop:resetCc}. From Theorem~\ref{thm:cc} the result
now follows immediately.\qedhere
\end{proof}
} 
\paperonly{
\begin{definition}
\label{def:reset}
Let $\PBES$ be a PBES, let $(V, \to)$ be its CFG, with labelling
$\markingsn{\syntactic}{}$\!. The PBES 
$\reset{\markingsn{\syntactic}{}}{\PBES}$ is obtained
from $\PBES$ by replacing every PVI 
$X(\val{e},\val{e'})$ in every $\rhs{X}$ of $\PBES$ by the formula
$\bigwedge_{\val{v} \in \values{\var{c}}} (\val{v} \not= \val{e} \vee
X(\val{e}, \resetvars{(X,\val{v})}{\markingsn{\syntactic}{}}{\val{e'}}))$.
The function $\resetvars{(X, \val{v})}{\markingsn{\syntactic}{}}{\val{e'}}$ is 
defined 
positionally as follows:
$$
\text{if  $\var[i]{d} \in \marking{\markingsn{}{}}{X,\val{v}}$ we set }
\ind{\resetvars{(X, \val{v})}{\markingsn{\syntactic}{}}{\val{e'}}}{i} = 
\val[\!\!i]{e'},
\text{ else }
\ind{\resetvars{(X, \val{v})}{\markingsn{\syntactic}{}}{\val{e'}}}{i} = 
\init{\var[i]{d}}.
$$
\end{definition}
Resetting dead parameters preserves the solution of the PBES, as we claim below.
\begin{restatable}{theorem}{resetSound}
\label{thm:resetSound}
Let $\PBES$ be a PBES, and 
$\markingsn{}{}$ a labelling. For all predicate variables $X$, and ground terms
$\val{v}$ and $\val{w}$:
$\sem{\PBES}{}{}(X(\sem{\val{v}},\sem{\val{w}})) =
\sem{\reset{\markingsn{\syntactic}{}}{\PBES}}{}{}(X(\sem{\val{v}},\sem{\val{w}}))$.
\end{restatable}
} 
As a consequence of the above theorem, instantiation of a PBES may become
feasible where this was not the case for the original PBES.
This is nicely illustrated by our running example, which now indeed
can be instantiated to a BES.
\begin{example}\label{exa:reset} Observe that parameter
$k^Z$ is not labelled in any of the $Z$ locations. This means that
$X$'s right-hand side essentially changes to:
$$
\begin{array}{c}
( i \not= 1 \vee j \not= 1 \vee X(2,j,k,l+1)) \wedge \\
\forall m\colon\sort{N}. (i \not= 1  \vee Z(i,2,1,k)) \wedge
\forall m\colon\sort{N}. (i \not= 2  \vee Z(i,2,1,k)) 
\end{array}
$$
Since variable $m$ no longer occurs in the above formula, the
quantifier can be eliminated. Applying the reset function on the
entire PBES leads to a PBES that we \emph{can} instantiate to a BES
(in contrast to the original PBES),
allowing us to compute that the
solution to $X(1,1,1,1)$ is $\true$. 
This BES has only 7 equations.

\end{example}

\section{Optimisation}\label{sec:local}

Constructing a CFG can suffer from a combinatorial explosion; \eg,
the size of the CFG underlying the following PBES 
is exponential in the number of detected CFPs. 
\[
\begin{array}{lcl}
\nu X(i_1,\dots,i_n \colon \sort{B}) & =&
(i_1 \wedge X(\false, \dots, i_n)) \vee
(\neg i_1 \wedge X(\true, \dots, i_n)) \vee \\
&\dots \vee 
& (i_n \wedge X(i_1, \dots, \false)) \vee
(\neg i_n \wedge X(i_1, \dots, \true))
\end{array}
\]
In this section we develop an alternative to the analysis of the
previous section which mitigates the combinatorial explosion
but still yields sound results. The 
correctness of our alternative is based on the following proposition,
which states that resetting using any labelling that approximates that of 
Def.~\ref{def:markingHeuristic} is sound.
\begin{proposition}\label{prop:approx}
  Let, for given PBES $\PBES$, $(\VCFGsyn, {\smash{\ECFGsyn}})$ be a 
  CFG with labelling $\markingsn{\syntactic}{}$, and let $L'$ be
  a labelling such that $\marking{\markingsn{\syntactic}{}}{X,\val{v}} \subseteq 
  L'(X,\val{v})$ for all $(X, \val{v})$. Then for all $X, \val{v}$ and
  $\val{w}$:
$\sem{\PBES}{}{}(X(\sem{\val{v}},\sem{\val{w}})) =
\sem{\reset{L'\!}{\PBES}}{}{}(X(\sem{\val{v}},\sem{\val{w}}))$
\end{proposition}
\reportonly{
\begin{proof}
  Let $(\VCFGsyn, {\smash{\ECFGsyn}})$ be a 
  CFG with labelling $\markingsn{\syntactic}{}$, and let $L'$ be
  a labelling such that $\marking{\markingsn{\syntactic}{}}{X,\val{v}}
  \subseteq L'(X,\val{v})$ for all $(X, \val{v})$.
  
  Define relation $R^{\resetname}_{\markingsn{\syntactic}{},L'}$ as follows.
  $$
  \begin{cases}
  X(\semval{v}, \semval{w})) R^{\resetname}_{L,L'} \changed{X}(\semval{v},
  \sem{\resetvars{(X, \val{v})}{L'}{\val{w})}}{}{}) \\
  X(\semval{v}, \semval{w}) R^{\resetname}_{L,L'} X(\semval{v}, \semval{w'}) & 
  \text{if }
  X(\semval{v}, \sem{\val{w}}{}{}) \markrel{\markingsn{\syntactic}{}} 
  X(\semval{v}, \semval{w'})
  \end{cases}
  $$
  
  The proof using $R^{\resetname}_{L,L'}$ now follows the exact same line of
  reasoning as the proof of Theorem~\ref{thm:resetSound}.\qedhere
\end{proof} 
}
The idea is to analyse a CFG consisting of disjoint subgraphs for
each individual CFP, where each
subgraph captures which PVIs are under the control of a CFP: only if the
CFP can confirm whether a predicate formula potentially depends on
a PVI, there will be an edge in the graph. 
As before, let
$\PBES$ be an arbitrary but fixed PBES, $(\sourcename, \destname,
\copiedname)$ a unicity constraint derived from $\PBES$, and 
$\var{c}$ a vector of CFPs.
\begin{definition}
\label{def:localCFGHeuristic} 
The \emph{local} control flow graph (LCFG) is a graph $(\VCFGloc, \ECFGloc)$ with:
\begin{compactitem}

\item $\VCFGloc = \{ (X, n, v) \mid X \in \bnd{\PBES} \land
n \le |\var{c}| \land v \in \values{\var[n]{c}} \}$, and

\item $\ECFGloc \subseteq \VCFGloc \times \mathbb{N} \times \VCFGloc$ is the 
least relation
satisfying $(X,n,v) \ECFGloc[i] (\predphi{\rhs{X}}{i},n,w)$ if:
\begin{compactitem}
  \item $\source{X}{i}{n} = v$ and $\dest{X}{i}{n} = w$, or
  \item $\source{X}{i}{n} = \bot$, $\predphi{\rhs{X}}{i} \neq X$ and 
        $\dest{X}{i}{n} = w$, or
  \item $\source{X}{i}{n} = \bot$, $\predphi{\rhs{X}}{i} \neq X$ and 
        $\copied{X}{i}{n} = n$ and $v = w$.
\end{compactitem}
\end{compactitem}
\end{definition}
We write $(X, n, v) \ECFGloc[i]$ if there exists some $(Y, m, w)$ such that
$(X, n, v) \ECFGloc[i] (Y, m, w)$.
Note that the size of an LCFG is $\mathcal{O}(|\bnd{\PBES}| \times |\var{c}|
\times \max\{ |\values{\var[k]{c}}| ~\mid~ 0 \leq k \leq |\var{c}| \})$.

We next describe how to label the LCFG in such a way that the
labelling meets the condition of Proposition~\ref{prop:approx},
ensuring soundness of our liveness analysis. The idea of using LCFGs
is that in practice, the use and alteration of a data parameter is entirely
determined by a single CFP, and that only on ``synchronisation points''
of two CFPs (when the values of the two CFPs are such that they
both confirm that a formula may depend on the same PVI) there is
exchange of information in the data parameters.  

 We first formalise when a
data parameter is involved in a recursion (\ie, when the parameter may
affect whether a formula depends on a PVI, or when a PVI may modify
the data parameter through a self-dependency or uses it to change another parameter).
Let $X \in \bnd{\PBES}$ be an arbitrary bound predicate variable in the
PBES $\PBES$.
\begin{definition}
\label{def:used} 
\label{def:changed} Denote  $\predinstphi{\rhs{X}}{i}$ by $Y(\var{e})$.
Parameter $\var[j]{d} \in \param{X}$ is:
\begin{compactitem}
\item \emph{used for} 
$Y(\var{e})$
if $\var[j]{d} \in \free{\guard{i}{\rhs{X}}}$;
\item \emph{used in}
$Y(\var{e})$ if for some $k$, we have $\var[j]{d} \in \free{\var[k]{e}}$,
 ($k \not=j$ if  
$X =Y$)
;
 
\item
\emph{changed} by 
$Y(\var{e})$ if both $X = Y$ and
$\var[j]{d} \neq \var[j]{e}$.
\end{compactitem}
%
\end{definition}
We say that a data parameter \emph{belongs to} a CFP if it controls
its complete dataflow.
\begin{definition}
\label{def:belongsTo}
\label{def:rules} 
CFP $\var[j]{c}$ \emph{rules}
$\predinstphi{\rhs{X}}{i}$ if $(X, j, v) \ECFGloc[i]$ for some $v$.
Let $d \in \param{X} \cap \varset{D}^{\DP}$ be a data parameter;
$d$ \emph{belongs to} $\var[j]{c}$ if and only if:
\begin{compactitem}
  \item whenever $d$ is 
        used for \emph{or} in $\predinstphi{\rhs{X}}{i}$, $\var[j]{c}$ rules
        $\predinstphi{\rhs{X}}{i}$, and
  \item whenever $d$ is changed by $\predinstphi{\rhs{X}}{i}$,
        $\var[j]{c}$ rules $\predinstphi{\rhs{X}}{i}$. 
\end{compactitem}
The set of data parameters that belong to $\var[j]{c}$ is denoted 
by $\belongsto{\var[j]{c}}$.
\end{definition}
By adding dummy CFPs that can only take on one value, we can ensure that
every data parameter belongs to at least one CFP.
For simplicity and without loss of generality, we can therefore 
continue to work under the following assumption.
\begin{assumption}\label{ass:belongs}
%
Each data parameter in an equation
belongs to at least one CFP.
\end{assumption}
We next describe how to conduct the liveness analysis using the
LCFG.  Every live data parameter is only labelled in those subgraphs
corresponding to the CFPs to which it belongs. The labelling
itself is constructed in much the same way as was done in the previous
section. 
Our base case labels a vertex $(X, n, v)$ with those parameters
that belong to the CFP and that are significant in $\rhs{X}$ when 
$\var[n]{c}$ has value $v$. 
The backwards reachability now dinstinguishes two cases,
based on whether the influence on live variables is internal to the CFP
or via an external CFP.
\begin{definition}
\label{def:relevanceLocalHeuristic}
Let $(\VCFGloc\!, \ECFGloc)$ be a LCFG for PBES
$\PBES$. The labelling $\markingsn{\local}{} \colon \VCFGloc 
\to \mathbb{P}(\varset{D}^{\DP})$ 
is defined as 
$\marking{\markingsn{\local}{}}{X, n, v}  = \bigcup_{k \in \N} \marking{\markingsn{\local}{k}}{X, n, v}$,
with
$\markingsn{\local}{k}$ inductively defined as:
\[
\begin{array}{ll}
\marking{\markingsn{\local}{0}}{X, n, v} & = 
\{ d \in \belongsto{\var[n]{c}} \mid d \in 
\significant{\simplify{\rhs{X}[\var[n]{c} := v]}} \} \\

\marking{\markingsn{\local}{k+1}}{X, n, v} & =
      \marking{\markingsn{\local}{k}}{X, n, v} \\
      
    & \quad \cup \{ d \in \belongsto{\var[n]{c}} \mid
      \exists i,w \text{ such that }~ \exists \var[\!\!\!\!\ell]{{d}^{Y}} \in 
      \marking{\markingsn{\local}{k}}{Y,n,w}: \\
          
    & \qquad (X, n, v) \ECFGloc[i] (Y, n, w)  \land 
    \affects{d}{\dataphi[\ell]{\rhs{X}}{i}} \} \\
                    
    & \quad \cup \{ d \in \belongsto{\var[n]{c}} \mid
    \exists i,m,v',w' \text{ such that }  (X, n, v) \ECFGloc[i] \\
    & \qquad \land\ \exists \var[\!\!\!\!\ell]{d^Y} \in 
    \marking{\markingsn{\local}{k}}{Y, m, w'}: \var[\!\!\!\!\ell]{d^Y} \not \in 
    \belongsto{\var[n]{c}} \\
    & \qquad \land\ (X,m,v') \ECFGloc[i] (Y,m,w') \land
       \affects{d}{\dataphi[\ell]{\rhs{X}}{i}} \}
                  
\end{array}
\]
\end{definition}
On top of this labelling we define the induced labelling
$\marking{\markingsn{\local}{}}{X, \val{v}}$, defined as $d \in
\marking{\markingsn{\local}{}}{X, \val{v}}$ iff for all $k$ for
which $d \in \belongsto{\var[k]{c}}$ we have $d \in
\marking{\markingsn{\local}{}}{X, k, \val[k]{v}}$.  This labelling
over-approximates the labelling of Def.~\ref{def:markingHeuristic}; \ie, we
have $\marking{\markingsn{\syntactic}{}}{X,\val{v}} \subseteq
\marking{\markingsn{\local}{}}{X,\val{v}}$ for all $(X,\val{v})$.
\reportonly{
We formalise this in the following lemma.
\begin{lemma}
  Let, for given PBES $\PBES$, $(\VCFGsyn, {\smash{\ECFGsyn}})$ be a global 
  control flow graph with
  labelling $\markingsn{\syntactic}{}$, and let $(\VCFGloc, \ECFGloc)$ be
  a local control flow graph  with
  labelling $\markingsn{\local}{}$, that has been lifted to the global CFG. Then
  $\marking{\markingsn{\syntactic}{}}{X,\val{v}} \subseteq 
  \marking{\markingsn{\local}{}}{X,\val{v}}$ for all $(X,\val{v})$.
\end{lemma}
  \begin{proof}
    We prove the more general statement that for all natural numbers $n$ it 
    holds 
    that
    $\forall (X, \val{v}) \in \VCFGsyn, \forall d \in 
    \marking{\markingsn{\syntactic}{n}}{X,
      \val{v}}: (\forall j: d \in \belongsto{\var[j]{c}} \implies
    d \in \marking{\markingsn{\local}{n}}{X, j, \val[j]{v}})$. The lemma
    then is an immediate consequence.
    
    We proceed by induction on $n$.
    \begin{compactitem}
      \item $n = 0$. Let $(X, \val{v})$ and $d \in 
      \marking{\markingsn{\syntactic}{0}}{X, \val{v}}$
      be arbitrary. We need to show that $\forall j: d \in 
      \belongsto{\var[j]{c}}
      \implies d \in \marking{\markingsn{\local}{0}}{X, j, \val[j]{v}}$.
      
      Let $j$ be arbitrary such that $d \in \belongsto{\var[j]{c}}$.
      Since $d \in \marking{\markingsn{\syntactic}{0}}{X, \val{v}}$, by 
      definition
      $d \in \significant{\simplify{\rhs{X}[\var{c} := \val{v}]}}$, hence also
      $d \in \significant{\simplify{\rhs{X}[\var[j]{c} := 
          \val[j]{v}}]}$.
      Combined with the assumption that  $d \in \belongsto{\var[j]{c}}$,
      this gives us $d \in \marking{\markingsn{\local}{0}}{X, j, \val[j]{v}}$
      according to Definition~\ref{def:localCFGHeuristic}.
      
      \item $n = m + 1$. As induction hypothesis assume for all $(X, \val{v}) 
      \in V$:
      \begin{equation}\label{eq:IHlocalapprox}
      \forall d: d \in 
      \marking{\markingsn{\syntactic}{m}}{X, \val{v}}
      \implies (\forall j: d \in \belongsto{\var[j]{c}}
      \implies d \in \marking{\markingsn{\local}{m}}{X,j,\val[j]{v}}).
      \end{equation}
      
      Let $(X, \val{v})$ be arbitrary with
      $d \in \marking{\markingsn{\syntactic}{m+1}}{X, \val{v}}$. Also let $j$
      be arbitrary, and assume that $d \in \belongsto{\var[j]{c}}$.
      
      We show that $d \in 
      \marking{\markingsn{\local}{m+1}}{X,j,\val[j]{v}}$ by distinguishing the
      cases of Definition~\ref{def:markingHeuristic}.
      If $d \in \marking{\markingsn{\syntactic}{m}}{X, \val{v}}$ the
      result follows immediately from the induction hypothesis. For the second
      case, suppose there are $i \in \mathbb{N}$ and $(Y, \val{w}) \in V$ such 
      that $(X, \val{v}) \smash{\ECFGsyn[i]} (Y, \val{w})$, 
      also assume there is some $\var[\ell]{d} \in 
      \marking{\markingsn{\syntactic}{m}}{Y, \val{w}}$
      with $d \in \free{\dataphi[\ell]{\rhs{X}}{i}}$. Let
      $i$ and $\var[\ell]{d}$ be such, and observe that
      $Y = \predphi{\rhs{X}}{i}$ and $i \leq \npred{\rhs{X}}$.
      According to the induction hypothesis,
      $\forall k: \var[\ell]{d} \in \belongsto{\var[k]{c}}
      \implies \var[\ell]{d} \in 
      \marking{\markingsn{\local}{m}}{Y, k, 
        \val[k]{w}}$.
      We distinguish two cases.
      
      \begin{compactitem}
        \item $\var[\ell]{d}$ belongs to $\var[j]{c}$. According to 
        \eqref{eq:IHlocalapprox}, we know
        $\var[\ell]{d} \in 
        \marking{\markingsn{\local}{m}}{Y, j, 
          \val[j]{w}}$.  
        Since $d \in \free{\dataphi[\ell]{\rhs{X}}{i}}$,
        we only need to show that $(X, j, \val[j]{v}) \ECFGloc[i] 
        (Y, j, \val[j]{w})$.
        We distinguish the cases for $j$ from 
        Definition~\ref{def:globalCFGHeuristic}.
        \begin{compactitem}
          \item $\source{X}{i}{j} = \val[j]{v}$ and $\dest{X}{i}{j} = 
          \val[j]{w}$,
          then according to Definition~\ref{def:localCFGHeuristic} $(X, j, 
          \val[j]{v}) \ECFGloc[i] 
          (Y, j, \val[k]{w})$          .
          \item $\source{X}{i}{j} = \bot$, $\copied{X}{i}{j} = j$ and 
          $\val[j]{v} = \val[j]{w}$.
          In case $Y \neq X$ the edge exists locally, and
          we are done.
          Now suppose that $Y = X$. Then 
          $\predinstphi{\rhs{X}}{i}$
          is not ruled by $\var[j]{c}$. Furthermore, $\var[\ell]{d}$
          is changed in $\predinstphi{\rhs{X}}{i}$, hence 
          $\var[\ell]{d}$
          cannot belong to $\var[j]{c}$, which is a contradiction.
          
          \item $\source{X}{i}{j} = \bot$, $\copied{X}{i}{j} = \bot$ and
          $\dest{X}{i}{j} = \val[j]{w}$. This is completely analogous to
          the previous case.
        \end{compactitem}
        \item $\var[\ell]{d}$ does not belong to $\var[j]{c}$. 
        Recall that there must be some $\var[k]{c}$ such that
        $\var[\ell]{d}$ belongs to $\var[k]{c}$, and by assumption now
        $\var[\ell]{d}$ does not belong to $\var[j]{c}$. Then
        according to Definition~\ref{def:relevanceLocalHeuristic}, $d$ is
        marked in $\marking{\markingsn{\local}{m+1}}{X, j, \val[j]{v}}$,
        provided that $(X, j, \val[j]{v}) \ECFGloc[i]$ and $(X, k, v') 
        \ECFGloc[i] (Y, k, \val[k]{w})$ for some
        $v'$. Let $v' = \val[k]{v}$ and $w' = \val[j]{w}$, according to the 
        exact same reasoning as
        before, the existence of the edges $(X,j,\val[j]{v}) \ECFGloc[i] (Y, j, 
        \val[j]{w})$ and $(X, k, \val[k]{v}) \ECFGloc[i] (Y, 
        k, \val[k]{w})$ can be shown, completing the proof.\qedhere
      \end{compactitem}
    \end{compactitem}
  \end{proof}
} 
Combined with Prop.~\ref{prop:approx}, this leads to the following
theorem.
\begin{theorem}
We have
$\sem{\PBES}(X(\semval{v}, \semval{w})) = 
  \sem{\reset{\markingsn{\local}{}}{\PBES}}{}{}(\changed{X}(\sem{\val{v}},\sem{\val{w}}))$
for all
predicate variables $X$ and ground terms $\val{v}$ and $\val{w}$.

\end{theorem}
The induced labelling $\markingsn{\local}{}$ can remain
implicit; in an implementation, the labelling constructed
by Def.~\ref{def:relevanceLocalHeuristic} can be used directly, sidestepping a
combinatorial explosion.

\section{Case Studies}\label{sec:experiments}

We implemented our techniques 
in the tool \texttt{pbesstategraph} of the mCRL2
toolset~\cite{Cra+:13}. Here, we report on the tool's effectiveness
in simplifying the PBESs originating from model checking problems and
behavioural equivalence checking problems: we compare sizes of the BESs
underlying the original PBESs to those for the PBESs obtained after
running the tool \texttt{pbesparelm} (implementing the techniques
from~\cite{OWW:09}) and those for the PBESs obtained after running
our tool. Furthermore, we compare the total times needed for reducing the PBES,
instantiating it into a BES, and solving this BES.

\begin{table}[!ht]
  \small
  \caption{Sizes of the BESs underlying (1) the original PBESs, and the
    reduced PBESs using (2) 
    \texttt{pbesparelm}, (3) \texttt{pbesstategraph} (global)
    and (4) \texttt{pbesstategraph} (local).
    For the original PBES, we report the number of generated BES equations,
    and the time required for generating and
    solving the resulting BES. For the other PBESs, we state the total
    reduction in percentages (\ie, $100*(|original|-|reduced|)/|original|$),
    and the reduction of the times (in percentages, computed in the same way),
    where for times we additionally include the \texttt{pbesstategraph/parelm}
    running times.
    Verdict $\surd$ indicates the problem 
    has solution $\true$; $\times$ indicates it is $\false$.
  }
  \label{tab:results}
  \centering
\scriptsize
\begin{tabular}{lc@{\hspace{5pt}}|@{\hspace{5pt}}rrrr@{\hspace{5pt}}|@{\hspace{5pt}}rrrr@{\hspace{5pt}}|@{\hspace{5pt}}c@{}}
 & \multicolumn{1}{c@{\hspace{10pt}}}{} & \multicolumn{4}{c}{Sizes} & 
 \multicolumn{4}{c}{Times}&Verdict\\
 \cmidrule(r){3-6}
 \cmidrule(r){7-10}
 \cmidrule{11-11}\\[-1.5ex]

& \multicolumn{1}{c}{} & \multicolumn{1}{@{\hspace{5pt}}c}{Original} & 
\multicolumn{1}{c}{\texttt{parelm}} & 
\multicolumn{1}{c}{\texttt{st.graph}} & 
\multicolumn{1}{c@{\hspace{5pt}}}{\texttt{st.graph}} & 
\multicolumn{1}{c}{Original}
& \multicolumn{1}{c}{\texttt{parelm}} & \multicolumn{1}{c}{\texttt{st.graph}} & 
\multicolumn{1}{c@{\hspace{5pt}}}{\texttt{st.graph}} & \\ 
           &   \multicolumn{1}{c@{\hspace{10pt}}}{$|D|$}
&&& \multicolumn{1}{c}{\texttt{(global)}}  & 
\multicolumn{1}{c}{\texttt{(local)}} &&& \multicolumn{1}{c}{\texttt{(global)}}  
& 
\multicolumn{1}{c}{\texttt{(local)}}& \\
\\[-1ex]
\toprule
\\[-1ex]
 \multicolumn{4}{c}{Model Checking Problems} \\
 \cmidrule{1-4}  \\[-1ex]

\multicolumn{11}{l}{\textbf{No deadlock}}  \\[.5ex]
\emph{Onebit}   & $2$ & 81,921 & 86\% & 89\% & 89\% & 15.7 & 90\% & 85\% & 90\% & $\surd$ \\
                 & $4$ & 742,401 & 98\% & 99\% & 99\% & 188.5 & 99\% & 99\% & 99\% & $\surd$ \\
\emph{Hesselink} & $2$ & 540,737 & 100\% & 100\% & 100\% & 64.9 & 99\% & 95\% & 99\% & $\surd$ \\
                  & $3$ & 13,834,801 & 100\% & 100\% & 100\% & 2776.3 & 100\% & 100\% & 100\% & $\surd$ \\
\\[-1ex]
\multicolumn{11}{l}{\textbf{No spontaneous generation of messages}}  \\[.5ex]
\emph{Onebit}  & $2$ & 185,089 & 83\% & 88\% & 88\% & 36.4 & 87\% & 85\% & 88\% & $\surd$ \\
                & $4$ & 5,588,481 & 98\% & 99\% & 99\% & 1178.4 & 99\% & 99\% & 99\% & $\surd$ \\

\\[-1ex]
\multicolumn{11}{l}{\textbf{Messages that are read are inevitably sent}}  
\\[.5ex]
\emph{Onebit}  & $2$ & 153,985 & 63\% & 73\% & 73\% & 30.8 & 70\% & 62\% & 73\% & $\times$ \\
                & $4$ & 1,549,057 & 88\% & 92\% & 92\% & 369.6 & 89\% & 90\% & 92\% & $\times$ \\

\\[-1ex]
\multicolumn{11}{l}{\textbf{Messages can overtake one another}}  \\[.5ex]
\emph{Onebit} & $2$ & 164,353 & 63\% & 73\% & 70\% & 36.4 & 70\% & 67\% & 79\% & $\times$ \\
               & $4$ & 1,735,681 & 88\% & 92\% & 90\% & 332.0 & 88\% & 88\% & 90\% & $\times$ \\

\\[-1ex]
\multicolumn{11}{l}{\textbf{Values written to the register can be read}} 
\\[.5ex]
 \emph{Hesselink} & $2$ & 1,093,761 & 1\% & 92\% & 92\% & 132.8 & -3\% & 90\% & 91\% & $\surd$ \\
                 & $3$ & 27,876,961 & 1\% & 98\% & 98\% & 5362.9 & 25\% & 98\% & 99\% & $\surd$ \\
\\[-1ex]
 \multicolumn{4}{c}{Equivalence Checking Problems} \\
 \cmidrule{1-4}  \\[-1ex]

\multicolumn{11}{l}{\textbf{Branching bisimulation equivalence}} \\[.5ex]
\emph{ABP-CABP} & $2$ & 31,265 & 0\% & 3\% & 0\% & 3.9 & -4\% & -1880\% & -167\% & $\surd$ \\
                 & $4$ & 73,665 & 0\% & 5\% & 0\% & 8.7 & -7\% & -1410\% & -72\% & $\surd$ \\
\emph{Buf-Onebit} & $2$ & 844,033 & 16\% & 23\% & 23\% & 112.1 & 30\% & 28\% & 31\% & $\surd$ \\
                   & $4$ & 8,754,689 & 32\% & 44\% & 44\% & 1344.6 & 35\% & 44\% & 37\% & $\surd$ \\
\emph{Hesselink I-S} & $2$ & 21,062,529 & 0\% & 93\% & 93\% & 4133.6 & 0\% & 74\% & 91\% & $\times$ \\

\\[-1ex]
\multicolumn{11}{l}{\textbf{Weak bisimulation equivalence}}  \\[.5ex]
\emph{ABP-CABP} & $2$ & 50,713 & 2\% & 6\% & 2\% & 5.3 & 2\% & -1338\% & -136\% & $\surd$ \\
                 & $4$ & 117,337 & 3\% & 10\% & 3\% & 13.0 & 4\% & -862\% & -75\% & $\surd$ \\
\emph{Buf-Onebit} & $2$ & 966,897 & 27\% & 33\% & 33\% & 111.6 & 20\% & 29\% & 28\% & $\surd$ \\
                   & $4$ & 9,868,225 & 41\% & 51\% & 51\% & 1531.1 & 34\% & 49\% & 52\% & $\surd$ \\
\emph{Hesselink I-S} & $2$ & 29,868,273 & 4\% & 93\% & 93\% & 5171.7 & 7\% & 79\% & 94\% & $\times$ \\

\\[-1ex]
\bottomrule

\end{tabular}

\end{table}

Our cases are taken from the literature. We here present a selection of the
results. For the model checking problems,
we considered the \emph{Onebit} protocol, which is a complex sliding window
protocol, and Hesselink's handshake register~\cite{Hes:98}. 
Both protocols are parametric in the set of values that can be read
and written.  A selection of properties of varying complexity and
varying nesting degree, expressed in the data-enhanced modal
$\mu$-calculus are checked.\footnote{\reportonly{The formulae are contained in 
the appendix;}
\paperonly{The formulae are contained in \cite{KWW:13report};} here we 
use textual characterisations instead.}
For the behavioural equivalence checking problems, we considered a
number of communication protocols such as the \emph{Alternating Bit
Protocol} (ABP), the \emph{Concurrent Alternating Bit Protocol} (CABP),
a two-place buffer (Buf) and the aforementioned Onebit protocol. Moreover,
we compare an implementation of Hesselink's register to a specification
of the protocol that is correct with respect to trace equivalence (but
for which currently no PBES encoding exists) but not with respect to the
two types of behavioural equivalence checking problems we consider here:
branching bisimilarity and weak bisimilarity.

The experiments were performed on a 64-bit Linux machine with kernel
version 2.6.27, consisting of 28 Intel\textregistered\ Xeon\copyright\ E5520 
Processors running
at 2.27GHz, and 1TB of shared main memory. None of our experiments use
multi-core features. We used revision 12637 of the mCRL2 toolset,
and the complete scripts for our test setup are available at 
\url{https://github.com/jkeiren/pbesstategraph-experiments}.

The results are reported in Table~\ref{tab:results};
higher percentages mean better reductions/\-smaller 
runtimes.\reportonly{\footnote{The absolute sizes and times are included in the 
appendix.}}
The experiments confirm our technique can achieve as much as an
additional reduction of about 97\% over \texttt{pbesparelm}, see the
model checking and equivalence problems
 for Hesselink's
register. Compared to the sizes of the BESs underlying the original PBESs,
the reductions can be immense. Furthermore,
reducing the PBES using the local stategraph algorithm, instantiating, and
subsequently solving it is typically faster than using the global stategraph 
algorithm,
even when the reduction achieved by the first is less.
For the equivalence checking
cases, when no reduction is achieved the local version of stategraph sometimes
results in substantially larger running times than parelm, which in turn already
adds an overhead compared to the original; however, for the cases in which this
happens the original running time is around or below 10 seconds, so the
observed increase may be due to inaccuracies in measuring.

\section{Conclusions and Future Work}\label{sec:conclusions}
We described a static analysis technique for PBESs that uses
a notion of control flow to determine when data parameters become
irrelevant.  Using this information, the PBES can be simplified, leading
to smaller underlying BESs. Our static analysis technique
enables the solving of PBESs using instantiation that so far could not be solved
this way as shown by our running example. 
Compared to existing techniques, our new static analysis technique can lead to
additional reductions of up-to 97\% in practical cases, as illustrated by our
experiments. Furthermore, if a reduction can be achieved the technique can
significantly speed up instantiation and solving, and in case no reduction is
possible, it typically does not negatively impact the total running time.

Several techniques described in this paper can be used
to enhance existing reduction techniques for PBESs. For instance,
our notion of a \emph{guard} of a predicate variable instance
in a PBES can be put to use to cheaply improve on the heuristics for
constant elimination~\cite{OWW:09}.  Moreover, we believe that our
(re)construction of control flow graphs from PBESs can be used to
automatically generate invariants for PBESs.  The theory on invariants
for PBESs is well-established, but still lacks proper
tool support.

\bibliographystyle{plain}
\bibliography{references}

\ifreport
\newpage
\appendix

\section{$\mu$-calculus formulae}\label{app:experiments}

Below, we list the formulae that were verified in
Section~\ref{sec:experiments}. All formulae are denoted in
the the first order modal $\mu$-calculus, an mCRL2-native data
extension of the modal $\mu$-calculus. The formulae assume that
there is a data specification defining a non-empty sort $D$ of
messages, and a set of parameterised actions that are present
in the protocols.  The scripts we used to generate our results,
and the complete data of the experiments are available from
\url{https://github.com/jkeiren/pbesstategraph}-\url{experiments}

\subsection{Onebit protocol verification}

\begin{itemize}

\item No deadlock:
\[
\nu X. [\true]X \wedge \langle \true \rangle \true
\]
Invariantly, over all reachable states at least one action is enabled.

\item Messages that are read are inevitably sent:
\[
\nu X. [\true]X \wedge \forall d\colon D.[ra(d)]\mu Y.([\overline{sb(d}]Y \wedge \langle \true \rangle \true))
\]
The protocol receives messages via action $ra$ and tries to send these 
to the other party. The other party can receive these via action $sb$.

\item Messages can be overtaken by other messages:
\[
\begin{array}{ll}
\mu X. \langle \true \rangle X \vee 
\exists d\colon D. \langle ra(d) \rangle
 \mu Y. \\
\qquad \Big ( \langle \overline{sb(d)}Y \vee \exists d'\colon D. d \neq d' \wedge
 \langle ra(d') \rangle \mu Z.  \\
\qquad \qquad ( \langle \overline{sb(d)} \rangle Z \vee
 \langle sb(d') \rangle \true) \\
\qquad \Big )
\end{array}
\]
That is, there is a trace in which message $d$ is read, and is still in the
protocol when another message $d'$ is read, which then is sent to the receiving
party before message $d$.

\item No spontaneous messages are generated:
\[
\begin{array}{ll}
\nu X.
[\overline{\exists d\colon D. ra(d)}]X \wedge\\
\qquad \forall d':D.  [ra(d')]\nu Y(m_1\colon D = d'). \\
  \qquad\qquad \Big ( [\overline{\exists d:D. ra(d) \vee sb(d) }]Y(m_1) \wedge \\
  \qquad\qquad\quad  \forall e\colon D.[sb(e)]((m_1 = e) \wedge X) \wedge \\ 
  \qquad\qquad\quad  \forall e':D. [ra(e')]\nu Z(m_2\colon D = e'). \\
  \qquad\qquad\qquad \Big ([\overline{\exists d \colon D. ra(d) \vee sb(d)}]Z(m_2) \wedge \\
  \qquad\qquad\qquad\quad       \forall f:D. [sb(f)]((f = m_1) \wedge Y(m_2))\\
  \qquad\qquad\qquad\quad       \Big )\\
  \qquad\qquad \Big )
\end{array}
\]
Since the onebit protocol can contain two messages at a time, the
formula states that only messages that are received can be subsequently
sent again. This requires storing messages that are currently in the buffer
using parameters $m_1$ and $m_2$.

\end{itemize}

\subsection{Hesselink's register}

\begin{itemize}
\item No deadlock:
\[\nu X. [\true]X \wedge \langle \true \rangle \true \]

\item Values that are written to the register can be read from the
register if no other value is written to the register in the meantime.
\[
\begin{array}{l}
\nu X. [\true]X \wedge
  \forall w \colon D . [begin\_write(w)]\nu Y.\\
\qquad\qquad \Big ( [\overline{end\_write}]Y \wedge 
    [end\_write]\nu Z. \\
\qquad\qquad\qquad \Big( [\overline{\exists d:D.begin\_write(d)}]Z \wedge
    [begin\_read]\nu W.\\
\qquad\qquad\qquad\qquad ([\overline{\exists d:D.begin\_write(d)}]W \wedge \\
\qquad\qquad\qquad\qquad\qquad \forall w': D . [end\_read(w')](w = w') ) \\
\qquad\qquad\qquad \Big) \\
\qquad\qquad \Big) \\
\end{array}
\]

\end{itemize}

\newpage
\section{Absolute sizes and times for the experiments}

\begin{table}[!ht]
  \small
  \caption{Sizes of the BESs underlying (1) the original PBESs, and the
    reduced PBESs using (2) 
    \texttt{pbesparelm}, (3) \texttt{pbesstategraph} (global)
    and (4) \texttt{pbesstategraph} (local).
    For each PBES, we report the number of generated BES equations,
    and the time required for generating and
    solving the resulting BES. For the other PBESs, we additionally include the 
    \texttt{pbesstategraph/parelm}
    running times.
    Verdict $\surd$ indicates the problem 
    has solution $\true$; $\times$ indicates it is $\false$.
  }
  \label{tab:results_absolute}
  \centering
\scriptsize
\begin{tabular}{lc@{\hspace{5pt}}|@{\hspace{5pt}}rrrr@{\hspace{5pt}}|@{\hspace{5pt}}rrrr@{\hspace{5pt}}|@{\hspace{5pt}}c@{}} &  & \multicolumn{4}{c}{Sizes} & \multicolumn{4}{c}{Times}&\\ & $|D|$  & Original & \texttt{parelm} & \texttt{st.graph} & \texttt{st.graph} & Original & \texttt{parelm} & \texttt{st.graph} & \texttt{st.graph} & verdict\\             &   &&& \texttt{(global)}  & \texttt{(local)} &&& \texttt{(global)}  & \texttt{(local)}& \\
\\[-1ex]
\toprule
\\[-1ex]
 \multicolumn{4}{c}{Model Checking Problems} \\
 \cmidrule{1-4}  \\[-1ex]

\multicolumn{11}{l}{\textbf{No deadlock}}  \\
\emph{Onebit}   & $2$ & 81,921 & 11,409 & 9,089 & 9,089 & 15.7 & 1.6 & 2.3 & 1.6 & $\surd$ \\
                 & $4$ & 742,401 & 11,409 & 9,089 & 9,089 & 188.5 & 1.6 & 2.3 & 1.3 & $\surd$ \\
\emph{Hesselink} & $2$ & 540,737 & 2,065 & 2,065 & 2,065 & 64.9 & 0.3 & 3.4 & 0.5 & $\surd$ \\
                  & $3$ & 13,834,801 & 2,065 & 2,065 & 2,065 & 2776.3 & 0.4 & 2.8 & 0.5 & $\surd$ \\
\\[-1ex]
\multicolumn{11}{l}{\textbf{No spontaneous generation of messages}}  \\
\emph{Onebit}  & $2$ & 185,089 & 30,593 & 22,145 & 22,145 & 36.4 & 4.7 & 5.6 & 4.4 & $\surd$ \\
                & $4$ & 5,588,481 & 92,289 & 60,161 & 60,161 & 1178.4 & 16.9 & 13.6 & 9.6 & $\surd$ \\

\\[-1ex]
\multicolumn{11}{l}{\textbf{Messages that are read are inevitably sent}}  \\
\emph{Onebit}  & $2$ & 153,985 & 57,553 & 41,473 & 41,473 & 30.8 & 9.1 & 11.8 & 8.2 & $\times$ \\
                & $4$ & 1,549,057 & 192,865 & 127,233 & 127,233 & 369.6 & 42.0 & 35.6 & 30.0 & $\times$ \\

\\[-1ex]
\multicolumn{11}{l}{\textbf{Messages can overtake one another}}  \\
\emph{Onebit} & $2$ & 164,353 & 61,441 & 44,609 & 49,217 & 36.4 & 11.0 & 11.9 & 7.6 & $\times$ \\
               & $4$ & 1,735,681 & 216,193 & 146,049 & 173,697 & 332.0 & 38.7 & 39.8 & 33.1 & $\times$ \\

\\[-1ex]
\multicolumn{11}{l}{\textbf{Values written to the register can be read}} \\
 \emph{Hesselink} & $2$ & 1,093,761 & 1,081,345 & 83,713 & 89,089 & 132.8 & 137.3 & 12.7 & 12.0 & $\surd$ \\
                 & $3$ & 27,876,961 & 27,656,641 & 528,769 & 561,169 & 5362.9 & 3995.5 & 81.3 & 72.2 & $\surd$ \\
\\[-1ex]
 \multicolumn{4}{c}{Equivalence Checking Problems} \\
 \cmidrule{1-4}  \\[-1ex]

\multicolumn{11}{l}{\textbf{Branching bisimulation equivalence}} \\
\emph{ABP-CABP} & $2$ & 31,265 & 31,265 & 30,225 & 31,265 & 3.9 & 4.0 & 76.4 & 10.3 & $\surd$ \\
                 & $4$ & 73,665 & 73,665 & 69,681 & 73,665 & 8.7 & 9.2 & 130.6 & 14.8 & $\surd$ \\
\emph{Buf-Onebit} & $2$ & 844,033 & 706,561 & 647,425 & 647,425 & 112.1 & 78.9 & 81.0 & 77.8 & $\surd$ \\
                   & $4$ & 8,754,689 & 5,939,201 & 4,897,793 & 4,897,793 & 1344.6 & 878.1 & 748.1 & 843.6 & $\surd$ \\
\emph{Hesselink I-S} & $2$ & 21,062,529 & 21,062,529 & 1,499,713 & 1,499,713 & 4133.6 & 4122.5 & 1070.6 & 375.3 & $\times$ \\

\\[-1ex]
\multicolumn{11}{l}{\textbf{Weak bisimulation equivalence}}  \\
\emph{ABP-CABP} & $2$ & 50,713 & 49,617 & 47,481 & 49,617 & 5.3 & 5.2 & 76.8 & 12.6 & $\surd$ \\
                 & $4$ & 117,337 & 113,361 & 106,089 & 113,361 & 13.0 & 12.5 & 125.3 & 22.8 & $\surd$ \\
\emph{Buf-Onebit} & $2$ & 966,897 & 706,033 & 644,209 & 644,209 & 111.6 & 89.6 & 79.8 & 80.6 & $\surd$ \\
                   & $4$ & 9,868,225 & 5,869,505 & 4,798,145 & 4,798,145 & 1531.1 & 1011.5 & 774.0 & 729.9 & $\surd$ \\
\emph{Hesselink I-S} & $2$ & 29,868,273 & 28,579,137 & 2,067,649 & 2,113,889 & 5171.7 & 4784.8 & 1061.3 & 294.1 & $\times$ \\

\\[-1ex]
\bottomrule

\end{tabular}

\end{table}

\fi

\end{document}